\documentclass[review,times]{elsarticle} 
\usepackage[utf8]{inputenc}
 \usepackage{amssymb}
 \usepackage{amsmath}
 \usepackage{amsthm}
 \usepackage{mathtools}
\usepackage{amsfonts} 
 \usepackage{setspace}
  \usepackage{mathrsfs}
\usepackage{geometry}
\usepackage{subfig}
\usepackage[usenames,dvipsnames]{xcolor}
\usepackage[pageanchor=false,colorlinks=true,linkcolor=blue,citecolor=blue,urlcolor=blue]{hyperref}

\newcommand\ti[1]{\tilde{#1}}

\newcommand\bR{\mathbb{R}}
\newcommand\bN{\mathbb{N}}
\newcommand\bRp{\bR_{+}}

\newcommand\sN{\mathbb{N}}

\newcommand\sL{\mathcal{L}}

\newcommand\sH{\mathcal{H}}
\newcommand\sA{\mathcal{A}}
\newcommand\sC{\mathcal{C}}
\newcommand\sE{\mathcal{E}}

\newcommand\sS{\mathcal{S}}
\newcommand\sB{\mathcal{B}}
\newcommand\sY{\mathcal{Y}}
\newcommand\sF{\mathcal{F}}

\newcommand\sT{\mathcal{T}}
\newcommand\mM{M}
\newcommand\vx{x}
\newcommand\vy{y}
\newcommand\vz{z}
\newcommand\ve{e}

\newcommand\vzero{0}
\newcommand\va{a}
\newcommand\vh{h}
\newcommand\vu{u}
\newcommand\vv{v}

\newcommand\QOC{\rho}

\newcommand\cl[1]{\overline{#1}}
\newcommand{\tam}{\mathrm{argmin}}

\newcommand{\tlim}{\mathrm{lim}}

\newcommand{\tconv}{\mathrm{conv}}
\newcommand{\tmin}{\mathrm{min}}
\newcommand{\tmax}{\mathrm{max}}

\newcommand{\tker}{\mathrm{ker}}

\newcommand{\ls}{\langle}
\newcommand{\rs}{\rangle}
\newcommand{\re}{\mathcal{R}e}

\newcommand{\defin}{:=}
\newcommand{\supp}{\textrm{supp}}
\newcommand{\linspan}{\textrm{span}}
\newtheorem{theorem}{Theorem}[section]
\newtheorem{corollary}{Corollary}[section]
\newtheorem{proposition}{Proposition}[section]
\newtheorem{fact}{Fact}[section]
\newtheorem{lemma}{Lemma}[section]
\newtheorem{definition}{Definition}[section]
\newtheorem{remark}{Remark}[section]

\begin{document}

\begin{frontmatter}
\title{
Stable recovery of low-dimensional cones in Hilbert spaces: \\
One RIP to rule them all}
\author[label2]{Yann Traonmilin\corref{cor1}}
 \cortext[cor1]{Corresponding author}
 \ead{yann.traonmilin@gmail.com}

\author[label2]{Rémi Gribonval}


 \address[label2]{INRIA Rennes - Bretagne Atlantique,\\ 
          Campus de Beaulieu \\
          35042 Rennes Cedex, France}


\begin{abstract}
Many inverse problems in signal processing deal with the robust estimation of unknown data from underdetermined linear observations. Low-dimensional models, when combined with  appropriate regularizers, have been shown to be efficient at performing this task. Sparse models with the 1-norm or low rank models with the nuclear norm are examples of such successful combinations. Stable recovery guarantees in these settings have been established using a common tool adapted to each case: the notion of restricted isometry property (RIP). In this paper, we establish generic RIP-based guarantees for the stable recovery of cones (positively homogeneous model sets) with arbitrary regularizers. These guarantees are illustrated on selected examples. For block structured sparsity in the infinite-dimensional setting, we use the guarantees for a family of regularizers which efficiency in terms of RIP constant can be controlled, leading to stronger and sharper guarantees than the state of the art.
\end{abstract}
\end{frontmatter}

{This work was partly funded  by  the  European  Research  Council,
PLEASE project (ERC-StG-2011-277906).}


\section{Introduction}

Linear inverse problems have become ubiquitous in signal processing and machine learning, in particular in the most extreme form of underdetermined linear inverse problems. In such settings, where some unknown data is to be estimated from underdetermined observations, a low-dimensional model is needed to complete these observations with some prior knowledge on the ``regularity'' of the unknown. Sparse and low-rank models have become particularly popular and provide a powerful theoretical and algorithmic framework to perform this estimation when the unknown is well approximated by a sparse vector or a low-rank matrix. 

\paragraph{Low-dimensional model sets}
A series of recent advances in low-dimensional regularization of inverse problems have established that stable reconstruction guarantees are also possible with more general low complexity models (beyond sparsity and low-rank) which provide the flexibility to model many types of data. Consider underdetermined observations $\vy = \mM \vx$ of a signal $\vx \in \sH$ ($\sH$ is a Hilbert space, finite or infinite-dimensional; $\vx$ may be a vector, a matrix, a tensor, a function \ldots) by a linear operator $\mM$ taking values in another Hilbert space $\sF$. The ``regularity'' or ``low-complexity'' of $\vx$ can be defined by the fact that $\vx \in \Sigma$, where $\Sigma \subset \sH$ is a {\em model set}. The two most traditional model sets encoding the low complexity of $\vx$ are the set of sparse vectors with at most $K$ non-zero coefficients, and the set of low-rank matrices with at most $r$ non-zero singular values. Many more model sets incorporating various other structures have been considered, most of them being unions of low-dimensional subspaces \cite{Blumensath_2009,Eldar_2009}. 

This paper focuses on general model sets with a special emphasis on models that are either {\em homogeneous ($t \cdot \vz \in \Sigma$ for any $t \in \bR$ and $\vz \in \Sigma$) or positively homogeneous ($t \cdot \vz \in \Sigma$ for any $t\geq0$ and $\vz \in \Sigma$)}. Homogeneous model sets are in fact unions of (possibly infinitely many) subspaces, and will be referred to as {\em UoS}. Positively homogeneous models sets are called {\em cones} and cover, e.g., the set of low-rank positive semi-definite matrices, the positive orthant, etc..

\paragraph{Uniform recovery guarantees and the RIP} Given a model set $\Sigma$, a widely studied fundamental question is the construction of decoders (i.e., functions taking the observations $y$ as input and yielding an estimation of $\vx$ as output) that recovers any $\vx \in \Sigma$ from observations $\vy=\mM\vx$ in the noiseless case,   
with stability guarantees in the presence of noise, and robustness when $\vx$ is only approximately in $\Sigma$. Such decoders are called \emph{instance-optimal}, see Section~\ref{sec:instance_optimality} for details.
Many efficient instance-optimal algorithms have been exhibited for classical sparse/low-rank model sets under appropriate restricted isometry properties (RIP). One can extend the notion of RIP of a linear operator $\mM$ to a general model set $\Sigma$ as follows (the set $\Sigma-\Sigma \defin \{\vx-\vx': \vx,\vx' \in \Sigma\}$ is called the {\em secant set} of $\Sigma$):
\begin{definition}[RIP]
$\mM: \sH \to \sF$ has the RIP on $\Sigma-\Sigma \subset \sH$ with constant $\delta$ if for all $\vx \in \Sigma -\Sigma$:
\begin{equation}\label{eq:DefRIP}
(1-\delta)\|\vx\|_\sH^2 \leq \|\mM \vx\|_\sF^2 \leq (1+\delta)\|\vx\|_\sH^2
\end{equation}
where $\|\cdot\|_{\sH}$ and $\|\cdot\|_{\sF}$ are the Euclidean norms of $\sH$ and $\sF$. 
\end{definition}

It is now well understood \cite{Bourrier_2014} that a lower RIP (the left hand side inequality in~\eqref{eq:DefRIP}) is necessary for the existence of instance-optimal decoders with respect to a given arbitrary model set $\Sigma$ in an arbitrary Banach space, and that a RIP is also sufficient for the existence of a decoder. However, such a decoder may not be easily computable.
Further, the existence of linear operators $\mM : \sH \to \bR^{m}$ satisfying the RIP on $\Sigma-\Sigma$ with an arbitrary model set $\Sigma \subset \sH$ has been established with general probabilistic arguments: $m$ need only be of the order of the {\em covering dimension of the normalized secant set} of $\Sigma$, which may be finite even when $\sH$ is infinite-dimensional  \cite{Eftekhari_2015,Dirksen_2014,Puy_2015}.

\paragraph{Decoding by constrained minimization}
For a selection of model sets $\Sigma$ and (often convex) regularizers $f$, a huge body of work \cite{Donoho_2006,Candes_2006b,Recht_2010, Candes_2010, Chandrasekaran_2012,Foucart_2013} has established that a RIP with small enough $\delta$ implies exact recovery guarantees for the particular decoder corresponding to the (convex) optimization program
\begin{equation} \label{eq:minimization}
  \underset{\vx \in \sH}{\tam}  f(\vx) \;s.t.\;  \mM \vx=\mM \vx_{0}.
\end{equation}
In the presence of bounded noise $\ve$ ($\|\ve\|_\sF \leq \eta$), decoders of the following form were shown to be instance-optimal:
\begin{equation} \label{eq:robust_minimization}
    \underset{\vx \in \sH}{\tam}   f(\vx) \;s.t.\;  \|\mM\vx-(\mM\vx_{0}+\ve)\|_\sF \leq \epsilon.
\end{equation}
A celebrated example \cite{Candes_2008} in $\sH = \bR^{n}$ is the combination of the $\ell^1$-norm decoder ($f(\vx) \defin \|\vx\|_{1}$) with Gaussian random projections ($\mM$ is $m \times n$ with i.i.d. Gaussian entries): when $m \gtrsim Klog(n/K)$, the matrix $\mM$ satisfies a RIP on the set $\Sigma-\Sigma$ of $2K$-sparse vectors that is sufficient to guarantee that $\ell^{1}$ minimization is instance-optimal on $K$-sparse vectors. 
Most RIP results aim at either generalizing the area of application \cite{Blumensath_2011,Candes_2011b} or the sharpness of hypotheses. For example the classical sufficient RIP condition $\delta < \sqrt{2}-1 \approx 0.414$ \cite{Candes_2008} for $\ell^{1}$ minimization was improved by several authors (see e.g. historical notes in \cite[Chapter 6]{Foucart_2013}). Cai and Zhang \cite{Cai_2014} closed this game for sparse vectors and low rank matrices by establishing recovery guarantees when $\delta < 1/\sqrt{2} \approx 0.707$, a result which is sharp in light of the work of Davies and Gribonval~\cite{Davies_2009}.

\paragraph{Convex regularization}

Convex regularizers are of practical interest. The decoder~\eqref{eq:robust_minimization} is indeed convex when $f$ is convex and most particularly when $f$ is a norm. Chandrasekaran et al. \cite{Chandrasekaran_2012}  gave general results for the recovery of sparse signals under random Gaussian observations (compressed sensing, random $M$) with atomic norms. These results give a unified understanding of a wide variety of norms used for low complexity recovery : $\ell^1$-norm, nuclear norm \ldots They give a way to calculate \emph{ non-uniform} recovery results by studying the relation between the geometry of the atomic norm and random observations. These results were extended to understand the phenomenon of phase transition~\cite{Ameluxen_2014} in compressed sensing. However such generic results are not available for \emph{uniform} recovery.

\paragraph{General objective} These observations lead us to the general objective of this paper. Our objective is to address the following question: \emph{given a general model set $\Sigma$ and a regularizer $f$, what RIP condition is sufficient to ensure that solving the optimization problem~\eqref{eq:robust_minimization} yields an instance-optimal decoder?} 

We summarize how we answer this question in this article in the following Section~\ref{sec:contributions}.


\subsection{Contributions and outline}\label{sec:contributions}

\paragraph{Admissible RIP constants for stable recovery with an arbitrary regularizer $f$} 
In this article, we give an admissible RIP constant $\delta_\Sigma(f)$ on $\Sigma -\Sigma$ that guarantees stable recovery of elements of $\Sigma$ with a regularizer $f$. As the notation suggests, the constant $\delta_\Sigma(f)$ only depends on the model $\Sigma$ and the regularizer $f$. It is related to the geometry of the model set $\Sigma$ through the \emph{atomic norm} $\|\cdot\|_\Sigma$ (a function which construction depends only on the model $\Sigma$, see Section~\ref{sec:atomicnorms}) and to the geometry of the  descent vectors of $f$ at points of $\Sigma$ (which union is denoted $\sT_f(\Sigma)$, its precise definition is given in Section~\ref{sec:definitions}). 
This constant is expressed as
\begin{equation}
 \delta_\Sigma(f) \defin \underset{\vz\in \sT_f(\Sigma) \setminus \{0\} }{\inf} \delta_\Sigma(\vz)
\end{equation}
where $ \delta_\Sigma(\vz) \defin \underset{\vx \in \Sigma}{\sup}\; \delta_{\Sigma}(\vx,\vz)$. Depending whether $\Sigma$ is a UoS or a cone:
\begin{eqnarray}
\delta_{\Sigma}(\vx,\vz) = \delta_\Sigma^{UoS}(\vx,\vz) &\defin&  \frac{-\re \ls \vx,\vz \rs}{\|\vx\|_\sH \sqrt{ \|\vx+\vz\|_\Sigma^2 -  \|\vx\|_\sH^2 - 2\re \ls \vx,\vz \rs}}\label{eq:DefDeltaUoS}\\
\delta_{\Sigma}(\vx,\vz) =  \delta^{cone}_\Sigma(\vx,\vz) &\defin& \frac{-2\re \ls \vx,\vz \rs}{ \|\vx+\vz\|_\Sigma^2- 2\re \ls \vx,\vz \rs}.\label{eq:DefDeltaCone}
\end{eqnarray}
We will show that in most conventional examples these constants can be bounded by estimates that are sharp with respect to a collection of models (e.g. $\delta_\Sigma(f) \geq 1/\sqrt{2}$ for sparse vectors and $\ell^1$-norm). Moreover, we can derive better than state-of-the art constants in other cases (block structured sparsity). For the cone case, we actually characterize a sharper constant with a more complicated expression, which will be exhibited precisely in Theorem~\ref{th:robust_RIP_eucl_precise}. 

Section~\ref{sec:theorems} is dedicated to the proof of our main result. This result  builds on the sharp RIP condition of Cai and Zhang for stable recovery of sparse vectors (respectively low rank matrices) by  $\ell^1$ norm (respectively nuclear norm) minimization, establishing that it can be extended to this general setting.
\begin{theorem}[RIP condition for exact recovery]\label{th:sharp_RIP}
Assume that $\Sigma$ is either a UoS or a cone. Then, for any $\mM$ continuous on $\sH$  that satisfies the RIP on $\Sigma-\Sigma$ with constant $\delta < \delta_\Sigma(f)$ we have: for all $\vx_{0} \in \Sigma$, $\vx_{0}$ is the unique minimizer of~\eqref{eq:minimization}.
\end{theorem}
When $\Sigma$ is not a cone (consider, e.g., a point cloud, see Section~\ref{sec:other_ex}), we can consider the cone generated by $\Sigma$, $\Sigma' \defin \bRp \Sigma$ instead. Since $\Sigma' \supset \Sigma$ is a cone, a RIP hypothesis on $\Sigma'-\Sigma'$ allows to apply the above Theorem and to provide recovery guarantees on $\Sigma$.

A stable recovery formulation for minimization~\eqref{eq:robust_minimization} is also established.

\begin{theorem}[RIP condition for stable recovery]\label{th:robust_RIP_eucl}
Assume that $\Sigma$ is either a UoS or a cone. Then, for any $\mM$ continuous on $\sH$  that satisfies the RIP on $\Sigma-\Sigma$ with constant $\delta < \delta_\Sigma(f)$ we have: for all $\vx_{0}\in\Sigma$, $\|\ve\|_\sF \leq \eta \leq \epsilon$, with $\vx^*$ the result of minimization~\eqref{eq:robust_minimization},
\begin{equation}
\| \vx^*-\vx_{0}\|_\sH \leq   C_{\Sigma}(f,\delta) \cdot (\eta +\epsilon )
\end{equation}
where $C_{\Sigma}(f,\delta) < +\infty$. 
\end{theorem}
The constant $C_{\Sigma}(f,\delta)$ is explicit in most classical examples (see also Section~\ref{sec:structured_sparsity} for an apparently new explicit value for block structured sparsity). These two theorems are direct consequences of the more technical Theorem~\ref{th:robust_RIP_eucl_precise} proven in Section~\ref{sec:theorems}.
As described in Section \ref{sec:instance_optimality}, the hypotheses of Theorem~\ref{th:robust_RIP_eucl} also imply robustness to modeling error (stable and robust instance optimality) for positively homogeneous, non-negative, convex functions  $f$. In these cases, the level of modeling error is not needed to set the parameter $\epsilon$. Such knowledge would be needed if the modeling error was considered as an observation noise. We will show on various classical examples how the additional hypotheses on $f$ are met.

\paragraph{Low complexity compressed sensing framework} These results give a general framework to check if a given regularizer $f$ is a good candidate to perform low complexity recovery and get a compressed sensing result when combining with the work of Dirsksen et al. \cite{Dirksen_2014} and Puy et al. \cite{Puy_2015}. Given a UoS (or a cone) $\Sigma$ and a function $f$, we can get a uniform recovery result for random observations of arbitrary signals in $\Sigma$ by establishing that (with high probability) a random linear measurement operator satisfies a RIP on the secant set $\Sigma-\Sigma$ . 

 Let $\Sigma$ be a model set and $f$ a candidate as a regularizer. We can follow these steps to get a potential low complexity compressed sensing result : 
\begin{enumerate}
\item Calculate (or bound) $\delta_\Sigma(f)$ and $C_\Sigma(f,\delta)$.
 \item Construct a dimension-reducing matrix $\mM$ satisfying the desired RIP on $\Sigma-\Sigma$, e.g., by applying techniques from \cite{Dirksen_2014} and \cite{Puy_2015} involving random subgaussian matrices. 
 \item Apply Theorem~\ref{th:robust_RIP_eucl} to conclude that regularization with $f$ leads to stable recovery of vectors from $\Sigma$ from their noisy low-dimensional linear measurements with $\mM$.
\end{enumerate}
This framework is illustrated on a few examples.
\begin{itemize}
\item 
{\bf Block Structured Sparsity in Infinite dimension.} 
We use our theorems in Section~\ref{sec:structured_sparsity} to give uniform recovery results under RIP conditions for the set $\Sigma$ of block group sparse vectors, using regularization with structured sparsity inducing norms (group norms, without overlap between groups).  
 
In a finite-dimensional setting, Ayaz et al. \cite{Ayaz_2014} established uniform recovery guarantees for compressed sensing of such structured sparse vectors under a RIP hypothesis. In this case, the regularizer is a mixed $\ell^1-\ell^2$ norm (also called group norm, it can be interpreted as an atomic norm, see Section~\ref{sec:atomic_norm_ex}). Given any matrix $M$ with a RIP constant $\delta <\sqrt{2}-1$ for vectors in $\Sigma - \Sigma$, this regularizer recovers vectors from $\Sigma$. 

Adcock and Hansen~\cite{Adcock_2012} propose a generalized sampling strategy to recover (structured) sparse signals in an infinite-dimensional setting. This setting is important to model the acquisition of analog signals, e.g. in magnetic resonance imaging (MRI). Bastounis et al.~\cite{Bastounis_2015} showed that the RIP constant must necessarily depend on the ratio of sparsity in each block and the number of blocks for the $\ell^1$ norm and block sparsity.

With our new general framework:
\begin{itemize}
\item we improve the RIP constant on group sparse vectors of Ayaz et al. \cite{Ayaz_2014} to the sharp $\frac{1}{\sqrt{2}}$;
\item we generalize the results of \cite{Ayaz_2014} to (non-overlapping) group norms and block sparsity in an {\em infinite-dimensional context} which encompasses the generalized sampling of Adcock and Hansen;
\item we exhibit a particular weighting of the group norm that removes the dependency of admissible RIP constant on the ratio of sparsity between blocks. 
\end{itemize}
\item
{\bf Other models and regularizers} In Section~\ref{sec:other_ex}, we show how some other classical examples can be handled by this framework (including {\em low rank matrix recovery}, where our results match the sharp RIP results of Cai and Zhang \cite{Cai_2014} on which they are built).

The theoretical derivations in this article highlight the role of a particular {\em atomic norm} $\|\cdot\|_\Sigma$ constructed directly from the model $\Sigma$ (the definition and properties of this norm can be found in Section~\ref{sec:definitions}). This norm can be found in the literature \cite{Argyriou_2012,Richard_2014} for some specific models but cannot be directly used as a regularizer with uniform recovery guarantees on $\Sigma$ in general. In Section~\ref{sec:sigma_norm}, we study $f=\|\cdot\|_\Sigma$ as a candidate for regularization:
\begin{itemize}
\item we first show that this norm enables stable recovery under a RIP assumption with constant $\delta< \frac{2(1-\mu(\Sigma)) }{ 3+2\mu(\Sigma) }$ when the cone model $\Sigma$ is an {\em arbitrary finite union of one-dimensional half-spaces} with coherence parameter $\mu(\Sigma)$ (see Section \ref{sec:one_dim_union}). 
\item we further show that the set $\Sigma$ of (non-negative multiples of) permutation matrices can be stably recovered with $O(nlog(n))$ subgaussian linear measurements by using the regularizer $\|\cdot\|_\Sigma$, the norm induced by the Birkhoff polytope of bi-stochastic matrices. This is the consequence of a sufficient RIP constant $\delta<\frac{2}{3}$ for recovery. Similar results were established regarding {\em non-uniform}\footnote{For every given matrix $\vx \in \Sigma$, the probability of drawing a (Gaussian) matrix $\mM$ yielding exact recovery of this particular $\vx$ is high.} recovery of permutation matrices \cite{Chandrasekaran_2012} with the Gaussian ensemble, with applications to the recovery of rankings from low-dimensional projections. Here we obtain {\em uniform}\footnote{The probability of drawing a (subgaussian) matrix $\mM$ yielding exact recovery of {\em every} $\vx \in \Sigma$ is high.} recovery guarantees with qualitatively as many linear measurements.
\item however, this norm cannot perform stable recovery in the general case. We show that the set $\Sigma$ of $K$-sparse vectors  {\em cannot be uniformly recovered} by regularization with $\|\cdot\|_\Sigma$ (called the $K$-support norm in \cite{Argyriou_2012}).
\end{itemize}
\end{itemize}

Finally, we discuss these results and conclude in Section~\ref{sec:discussion}, especially, future directions for the design of regularizers.

\section{Preliminaries and definitions}\label{sec:definitions}
We consider a real or complex Hilbert space $\sH$ equipped with a Hermitian inner product $\ls \cdot,\cdot\rs$ and a Euclidean norm  $\|\vx\|_\sH^2 \defin \ls \vx,\vx\rs$. The Pythagorean property $\|\vx+\vx'\|_\sH^2 = \|\vx\|_\sH^2 + 2\re \ls \vx,\vx'\rs  + \|\vx'\|_\sH^2$ will play a major role in the proofs. Convergence of sequences, of sums, closure of sets and continuity of functions are all defined with respect to $\|\cdot\|_\sH$. For $\gamma \in \bRp=\{a \in \bR : a \geq 0 \}$, we define the sphere and the closed ball of radius $\gamma$,
\begin{equation}
\begin{split}
 &\sS(\gamma)=\{x \in \sH : \|x\|_\sH = \gamma \}, \\
 &\sB(\gamma)=\{x \in \sH : \|x\|_\sH \leq \gamma \}.
\end{split}
\end{equation}
The model set is some subset $\Sigma \subset \sH$. 
 The operator $\mM$ is a linear map from $\sH$ to a Hilbert space $(\sF, \|\cdot\|_\sF)$,   thus $\|\cdot\|_\sF$ is as well Euclidean.
 Its domain $\sH_{\mM} \defin \textrm{dom}(\mM) \subset \sH$ may be strictly smaller than $\sH$.

 The regularizer $f$ is a generalized function from $\sH$ to $\bR \cup \{+\infty\}$ which is finite on the set $\sH_{f} \defin \textrm{dom}(f) \subset \sH$, and we assume that $\Sigma \subset \sH_{f} \cap \sH_{\mM}$ so that for $\vx \in \Sigma$, $f(\vx) < \infty$ and $\mM \vx$ is well-defined. 
 \subsection{Descent sets}
 The notion of descent vector plays an important role in the analysis conducted in this paper. 
 \begin{definition}[Descent vectors]
 For any $\vx \in \sH$, the collection of descent vectors of $f$ at $\vx$ is
 \begin{equation}
 \sT_{f}(\vx) \defin \left\{ \vz \in \sH : f(\vx+\vz) \leq f(\vx) \right\}
 \end{equation}
\end{definition}
In the noiseless case, the main question investigated in this paper is whether for all $\vx \in \Sigma$,
\begin{equation}
\label{eq:recovery}
  \{\vx\} = \underset{\tilde{\vx} \in \sH_{\mM}}{\tam}  f(\tilde{\vx})\; s.t.\;  \mM \tilde{\vx}=\mM \vx,
\end{equation}
i.e., the minimizer is unique and matches $\vx$.

This is easily seen to be equivalent to $\tker \mM \cap \sT_{f}(\vx) = \{\vzero\}$, for all $\vx \in \Sigma$, where $\tker \mM \defin \{\vz \in \sH_{\mM}, \mM \vz = \vzero\}$ is the null space of $\mM$. In other words, the property we wish to establish under a RIP assumption is
\begin{equation}
\tker \mM \cap \sT_f(\Sigma) = \{\vzero\}.
\end{equation}
where as a shorthand the union of all descent vectors on points of the model $\Sigma$, later also called descent set for brevity, is denoted 
\begin{equation}
\sT_f(\Sigma) \defin \underset{\vx\in \Sigma}{\bigcup} \sT_{f}(\vx).
\end{equation}

\subsection{Atomic norms}
\label{sec:atomicnorms}
A particular family of {\em convex} regularizers $f$ is defined using the notion of {\em atomic ``norm''} \cite{Chandrasekaran_2012} (which in fact is not always a norm). Considering a set $\sA \subset \sH$, commonly called the set of {\em atoms}, the corresponding {\em atomic ``norm''} is built using the convex hull of $\sA$.

\begin{definition}[Convex hull]
The convex hull of a set $\sA$ is : 
\begin{equation}
\tconv(\sA) \defin  \left\{\vx = \sum c_i \va_i : \va_i \in \sA, c_i \in \bRp, \sum c_i = 1\right\}
\end{equation}
\end{definition}

\begin{definition}[Atomic norm]
 The atomic ``norm'' induced by the set $\sA$ is defined as: 
\begin{equation}
 \|\vx\|_\sA \defin \inf \left\{ t \in \bRp:  \vx \in t\cdot\cl{ \tconv}(\sA) \right\}
\end{equation}
where $\cl{ \tconv}(\sA)$ is the closure of $\tconv(\sA)$ in $\sH$. This ``norm'' is finite only on 
\begin{equation}
\sE(\sA) \defin \bRp \cdot \cl{\tconv}(\sA) = \{\vx = t\cdot\vy, t \in \bRp,\vy \in \cl{\tconv}(\sA)\} \subset \sH.
\end{equation}
It is extended to $\sH$ by setting $\|\vx\|_\sA = +\infty$ if $\vx \notin \sE(\sA)$. 
\end{definition}

The main properties of $\|\cdot\|_\sA$ depend on the considered set $\sA$ through those of the convex set $\sC \defin \cl{\tconv(\sA)}$. 
In  fact, the atomic ``norm'' $\|\cdot\|_\sA$ is the gauge function associated with the convex set $\sC$ \cite{Rockafellar_2009}. Table~\ref{tab:norm_prop} summarizes the properties of such gauges.  The proof of these results in infinite dimension can be found in \cite{Charalambos_2006}, the finite dimension is covered by \cite{Rockafellar_1970}.

Moreover, from \cite{Bonsall_1991,Feichtinger_2002}, $\sE(\sA)$ is continuously embedded in $\sH$ if $\sA$ is {\em bounded}  (because  $\sE(\sA \cup -\sA )$ is continuously embedded in $\sH$). Thus we will suppose  $\sA$  bounded from now on. 

By abuse of notation, we define  $\sT_{\sA}(\vx) \defin \sT_{\|\cdot\|_{\sA}}(\vx)$ and $\sT_\sA(\Sigma) = \underset{\vx\in \Sigma}{\bigcup} \sT_{\sA}(\vx)$.

\begin{table}[!h]
\begin{tabular}{|p{0.5\textwidth}|p{0.5\textwidth}|}
\hline
 Hypotheses on $\sC \defin \cl{\tconv}(\sA)$ & Properties of $\|\cdot\|_\sA$ in $\sE(\sA) = \bRp \cdot \sC$ \\
 \hline
   No hypotheses        & Non-negative, lower semi-continuous and sublinear \\
 \hline
 $0 \in \sC$    & Positively homogeneous\\
 \hline
 $0 \in \text{int} \sC $ & Continuous \\
 \hline
$\sC$ is bounded (and closed by definition) & Coercive ( $\tlim_{\|x\|_\sH\to+\infty} \|x\|_\sA = + \infty)$\\
 \hline
 $-\sC= \sC$ & Is a norm\\
 \hline
\end{tabular}
\caption{Link between $\sA$ and $\|\cdot\|_\sA$ using the properties of gauges.}\label{tab:norm_prop}
\end{table}

\subsubsection{Some well known atomic norms}\label{sec:atomic_norm_ex} As pointed out in \cite{Chandrasekaran_2012}, many well know norms used for low complexity recovery are atomic norms with atoms of interest. Atoms are often \emph{normalized}. We say that a vector $u$ is \emph{normalized} if  $\|u\|_\sH=1$ :
\begin{itemize}
\item  $\ell^1$ norm : $\sA$ is the set of canonical orthonormal basis vectors multiplied by a complex scalar with modulus $1$, i.e. the normalized $1$-sparse vectors.
\item  Nuclear norm :  $\sA$ is the set of normalized rank one matrices (Section~\ref{sec:other_ex}).
\item Mixed $\ell^1-\ell^2$ norms (covered extensively in Section~\ref{sec:structured_sparsity}) : $\sA$ is the set of normalized $1$-group-sparse vectors. 
\item Gauge generated by a finite polytope : $\sA$ is composed of the vertices of a polytope (Section~\ref{sec:sigma_norm}). 
\item Spectral norm : $\sA$ is the set of normalized orthogonal matrices.
\end{itemize}

\subsubsection{Duality}
The dual of an atomic norm, defined as 
\begin{equation}
 \|\vx\|_\sA^*\defin \underset{\|\vu\|_\sA=1}{\sup} |\ls \vx,\vu \rs|
\end{equation}
is
\begin{equation}
 \|\vx\|_\sA^* = \underset{\va \in \sA}{\sup} |\ls \vx,\va \rs|.
\end{equation}
and satisfies the usual property $\|\vx\|_{\sH}^{2} \leq \|\vx\|_{\sA} \cdot \|\vx\|_{\sA}^{*}$ for all $\vx \in \sH$.

\subsubsection{Alternate characterization of atomic norms with normalized atoms}
The following bound, which is valid in particular when $\sA$ is made of {\em normalized atoms}, will be useful in the proof of our main results. 
\begin{fact}\label{fact:norm_formulation}
Suppose $\sA \subset \sB(1)$. We define
\begin{equation}
\|\vx\|_{\sA}^{'} \defin \inf \left\{ \sqrt{\sum \lambda_i \|\vu_i\|_\sH^2} :  \lambda_i \in \bRp, \sum \lambda_i =1, 
 \vu_i \in \bRp\sA, \vx=\sum  \lambda_i \vu_i\right\}.
\end{equation}
Then for any $\vx \in \sE(\sA)$ we have
\begin{equation}
 \|\vx\|_\sA \geq \underset{\gamma \to 0_+}{\lim} \;{\inf} \{ \|\tilde{\vx}\|_{\sA}^{'} : \tilde{\vx} \in t \cdot\tconv(\sA), \|\tilde{\vx}-\vx\|_{\sH} \leq \gamma\}.
\end{equation}
When $\tconv(\sA)$ is closed we have in fact $\|\vx\|_{\sA} = \|\vx\|_{\sA}^{'}$ for any $\vx$.
\end{fact}
\begin{proof}
By definition of $\tconv(\sA)$, if $\vx \in t \cdot\tconv(\sA)$ there are vectors $\va_i \in \sA$ and scalars $\lambda_{i} \in \bRp$ such that  $\vx = t\sum \lambda_i \va_i$ and $\sum_{i} \lambda_{i}=1$. Denoting $\vu_{i} \defin t \va_{i}$, we have $\vx=\sum_{i} \lambda_{i} \vu_{i}$ and $\vu_{i} \in \bRp\sA$ and we can indeed define $\|\vx\|_{\sA}^{'}$. 
Since $\sA \subset \sB(1)$ we have $\|\vu_{i}\|_{\sH} = \|t\va_i\|_\sH \leq t$, hence $\sum \lambda_i \|\vu_i\|_\sH^2 \leq t^2$, yielding $\|\vx\|_{\sA}^{'} \leq t$. Finally, observe that $\vu_{i} \in \bRp \sA$ implies $\|\vu_{i}\|_{\sA} \leq \|\vu_{i}\|_{\sH}$. Hence, by convexity of $\|\cdot\|_\sA^2$, for any $\lambda_{i} \in \bRp$ and $\vu_{i} \in \bRp \sA$ such that $\sum \lambda_{i}=1$ and $\sum \lambda_{i} \vu_{i} = \vx$ we thus have
\begin{equation}
 \|\vx\|_\sA^2 =  \|\sum_i  \lambda_i \vu_i\|_\sA^2 \leq \sum_i \lambda_i\| \vu_i\|_\sA^2 \leq \sum_i \lambda_i\| \vu_i\|_\sH^2 
\end{equation}
The infimum of the right side establishes that $\|\vx\|_{\sA} \leq \|\vx\|_{\sA}^{'}$.

Consider now an arbitrary $\vx \in \sH$, and $t > 0$ such that $\vx \in t \cdot \cl{\tconv}(\sA)$. By definition of $\cl{\tconv}(\sA)$, for any $\gamma > 0$ there is $\tilde{\vx} \in t \cdot\tconv(\sA)$ such that $\|\vx-\tilde{\vx}\|_{\sH} \leq \gamma$. 
As shown above, for any such $\gamma$ and $\tilde{\vx}$, since $\tilde{\vx} \in t \cdot\tconv(\sA)$ we have $t \geq \|\tilde{\vx}\|_{\sA}^{'}$, hence
\begin{equation}
t 
\geq 
\underset{\gamma>0}{\sup} \;  \inf \{ \|\tilde{\vx}\|_{\sA}^{'} : \tilde{\vx} \in t \cdot \tconv(\sA), \|\tilde{\vx}-\vx\|_{\sH} \leq \gamma\} = 
\underset{\gamma \to 0_{+}}{\lim} \;  \inf \{ \|\tilde{\vx}\|_{\sA}^{'} : \tilde{\vx} \in t\cdot \tconv(\sA), \|\tilde{\vx}-\vx\|_{\sH} \leq \gamma\}.
\end{equation}
This holds for any $t>0$ such that $\vx \in t\cdot \cl{\tconv}(\sA)$. We conclude by observing that $\|\vx\|_{\sA}$ is precisely defined as the infimum of such $t$.

When $\tconv(\sA)$ is closed, the above argument simplifies since $\vx \in t \cdot \cl{\tconv}(\sA)$ is equivalent to $\vx \in t\cdot \tconv(\sA)$, implying that $\|\vx\|_{\sA} \leq \|\vx\|_{\sA}^{'} \leq t$. Taking the infimum over such $t$ yields $\|\vx\|_{\sA} = \|\vx\|_{\sA}^{'}$.
\end{proof}

\subsubsection{Atomic norms associated to cones}
Given a cone $\Sigma \subset \sH$, the norm associated to {\em normalized atoms} 
\begin{equation}
\sA(\Sigma) \defin \Sigma \cap \sS(1)
\end{equation}
will be of particular interest for the RIP analysis conducted in this paper. As a shorthand we define:
\begin{equation}
\|\cdot\|_\Sigma \defin \|\cdot\|_{\Sigma \cap \sS(1)}.
\end{equation}
From the general properties of atomic norms (Table~\ref{tab:norm_prop}) we have:
\begin{fact}[The ``norm'' $\|\cdot\|_{\Sigma}$] Given any cone $\Sigma$ we have:
\begin{enumerate}
  \item For all $\vz\in \sH$, $\|\vz\|_\Sigma \geq \|\vz\|_\sH$.
 \item For all $\vx\in \Sigma$, $\|\vx\|_\Sigma =\|\vx\|_\sH$.
\item $\|\cdot\|_\Sigma$ is  non-negative, lower semi-continuous, sublinear, coercive and positively homogeneous.
 \item When $\Sigma$ is homogeneous (i.e., a UoS), $\|\cdot\|_\Sigma$ is indeed a norm (since $\Sigma = -\Sigma$).
\end{enumerate}
\end{fact}

\subsection{Atomic norms are often sufficient} \label{sec:atomic_norm_sufficient}
Ultimately, our aim is to find a convex regularizer that enables the uniform recovery of $\Sigma$ for some linear observation operator $\mM$. We show below that if such a convex regularizer exists, there exists an atomic norm with atoms $\sA \subset \Sigma$ (independent of $\mM$) that also ensures uniform recovery of $\Sigma$. This means that the search for a convex regularizer can be limited to the subset of such atomic norms.  In fact, Lemma~\ref{lem:atomic_necessary} below implies that we can restrict our study to such atomic norms, as soon as there exists some ``well-behaved'' regularizer $f$ with uniform recovery guarantees.
\begin{lemma}\label{lem:atomic_necessary}
Let $\Sigma$ be a cone and $f$ be a proper coercive continuous regularizer. Assume that for some $t > f(0)$, the level set $\sL(f,t)= \{y\in \sH : f(y)\leq t \}$ is a convex set. Then there exists a family of atoms $\sA \subset \Sigma$ such that : 
\begin{equation}
\sT_\sA(\Sigma) \subset \bRp^*. \sT_f(\Sigma)
\end{equation}
where $\bRp^*=\bRp \setminus \{0\}$.
\end{lemma}
\begin{proof}
We define $\sA \defin \sL(f,t) \cap \Sigma = \{\vx\in \Sigma : f(\vx) \leq t \}$. 
Let $\vz \in \sT_{\sA}(\Sigma)$. By definition, there exists $\vx\in \Sigma$ such that 
$ \|\vx+\vz\|_\sA \leq \|\vx\|_\sA $.  
We have $f(0 \cdot \vx) = f(0) < t$, and since $f$ is coercive $f(\lambda \vx)\underset{\lambda \to +\infty}{\to} +\infty$. Thus, by the continuity of $f$ there is $\lambda_0 >0$ such that  $ f(\lambda_0 \vx) = t$. Since $\Sigma$ is a cone, the vector $\vx' = \lambda_0 \vx$ belongs to $\Sigma$ and, since  $f(\vx')=t$, by definition of $\sA$ we have indeed $\vx' \in \sA$ and $\|\vx'\|_{\sA} \leq 1$. Furthermore, we have  $\|\vx'+ \lambda_0 \vz\|_\sA = \lambda_{0} \|\vx+\vz\|_{\sA} \leq \lambda_{0} \|\vx\|_{\sA} = \|\vx'\|_\sA$. 

We now observe that, on the one hand, the level set $\sL( \|\cdot\|_\sA,1) = \overline{\tconv}(\sA)$  is the smallest closed convex set containing $\sA$; on the other hand $\sA \subset \sL( f,t)$ and $\sL( f,t)$ is convex and closed (by the continuity of $f$). Thus $\sL( \|\cdot\|_\sA,1) \subset \sL(f,t)$ and the fact that $\|\vx'+\lambda_{0}\vz\|_{\sA} \leq \|\vx'\|_{\sA} \leq 1$ implies \begin{equation}
f(\vx'+\lambda_{0} \vz) \leq t = f(\vx').
\end{equation}
This shows that $\lambda_{0} \vz \in \sT_f(\Sigma)$ and finally that $\vz \in \bRp^*. \sT_f(\Sigma)$.

\end{proof}

\begin{corollary}
With such a set of atoms, we have  $\delta_{\Sigma}(\|\cdot\|_{\sA}) \geq \delta_{\Sigma}(f)$.
\end{corollary}
\begin{proof}
First remark that for all $\lambda>0$,  $\delta_\Sigma(\lambda\vz) = \delta_\Sigma(\vz)$ both in the UoS and cone cases. Thus, because $\sT_\sA(\Sigma) \subset \bRp^*. \sT_f(\Sigma)$, we have
\begin{equation}
 \delta_\Sigma  (\|\cdot\|_{\sA}) = \underset{z\in\sT_\sA(\Sigma) \setminus \{0\}}{\inf} \delta_\Sigma(z) \geq  \underset{z\in\bRp^*. \sT_f(\Sigma) }{\inf} \delta_\Sigma(z) =  \underset{z\in \sT_f(\Sigma) \setminus \{0\} }{\inf}\delta_\Sigma(z) = \delta_\Sigma  (f) 
\end{equation}
 
\end{proof}

In other words this corollary shows that given a sufficiently well behaved regularizer $f$ (in terms of recovery guarantees Theorem~\ref{th:robust_RIP_eucl} under a RIP assumption on $\Sigma-\Sigma$), we can find an atomic norm that behaves at least as well, possibly with an even less restrictive RIP condition.

\subsection{Tools to calculate admissible RIP constants }
In all examples, explicit bounds on $\delta_\Sigma(f)$ are achieved by decomposing descent vectors $\vz \in \sT_f(\Sigma)$ as $\vz = (\vx+\vz) -\vx$, with $\vx \in \Sigma$. The ``goodness'' of such a decomposition is a trade-off between: 
\begin{itemize}
\item the incoherence between the two components, measured by the {\em quasi-orthogonality constant} 
\begin{equation}
\label{eq:DefQOC}
\QOC(\vx,\vz) \defin \frac{\re \ls \vx,\vx+\vz \rs}{\|\vx\|_\sH^2};
\end{equation}
\item the fact that $\vz$ is (quasi) a descent vector for the $\|\cdot\|_{\Sigma}$ norm at $\vx \in \Sigma$, as measured by the {\em quasi-descent constant} 
\begin{equation}
\label{eq:DefQDC}
\alpha_{\Sigma}(\vx,\vz) \defin \frac{\|\vx+\vz\|_{\Sigma}^{2}}{\|\vx\|_{\Sigma}^{2}}.
\end{equation}
\end{itemize}
For most examples in Section~\ref{sec:structured_sparsity} and Section~\ref{sec:other_ex},  we bound $\delta_{\Sigma}(f)$ by exhibiting for each descent vector $\vz \in \sT_f(\Sigma) \setminus \{0\}$ an element of the model, $\vx \in \Sigma$, such that  $\QOC(\vx,\vz)=0$ and $\alpha_{\Sigma}(\vx,\vz) \leq \alpha_\Sigma(f)$, where $\alpha_{\Sigma}(f)$ is a constant that depends only on $f$ and $\Sigma$.  When this is possible we conclude that $\delta_\Sigma(f) \geq 1/\sqrt{1+\alpha_\Sigma(f)}$.

\section{Stable and robust recovery of cones under a RIP condition}\label{sec:theorems}
This section is dedicated to the precise statement and proof of our main result. Applications of these results will be discussed in later sections.

\subsection{Exact recovery and stability to noise}

The main part of the argument is a generalization of the sharp RIP result by Cai and Zhang \cite{Cai_2014} whose RIP theorem states that $\delta <\sqrt{1/2}$ on $\Sigma-\Sigma$ implies stable recovery of sparse vectors (when $\Sigma$ is the set of $K$-sparse vectors) with the $\ell^1$-norm, and of low rank matrices (when $\Sigma$ is the set of rank-$r$ matrices) with the nuclear norm. This RIP theorem is sharp because for any $\epsilon >0$ there exists a dimension $N$ of $\sH=\bR^N$, a sparsity model $\Sigma_K$ and matrix $\mM$ satisfying the RIP with constant $\delta \geq \sqrt{1/2} +\epsilon$ for which uniform recovery of by $\ell^{1}$ minimization is impossible \cite{Davies_2009}.

Theorem~\ref{th:sharp_RIP} and \ref{th:robust_RIP_eucl} are immediate corollaries of Theorem~\ref{th:robust_RIP_eucl_precise} below. The proof of Theorem~\ref{th:robust_RIP_eucl_precise} relies on the following identity, where we recall that $\|\cdot\|_\sF$ is the Euclidean norm. 
\begin{fact}\label{prop:ident0}
Let $\lambda_i \in \bRp$ such that $\sum \lambda_i =1$, and 
$\vh_i \in \sH$ such that the $\|\mM\vh_i\|_\sF$ are uniformly bounded.

Then 
 \begin{equation}\label{eq:ident0}
 \sum_i \lambda_i   \|\mM \vh_i \|_\sF^2=4\sum_i \lambda_i \| \big(\textstyle\sum_j \lambda_j\mM  \vh_j\big) -\tfrac{1}{2} \mM\vh_i \|_\sF^2.
\end{equation}
\end{fact}

\begin{proof}
The boundedness of the $\|\mM\vh_i\|_\sF$ imply that $\textstyle\sum_j \lambda_j \mM \vh_j$ and $\sum_i \lambda_i \|\mM\vh_i\|_\sF^2$ converge. Moreover, by convexity  $\| \textstyle\sum_j \lambda_j \mM \vh_j \|_\sF^2 \leq \sum_j \lambda_j\| \mM\vh_j \|_\sF^2 \leq \sup_j \| \mM \vh_j \|_\sF^2$. Hence  $\textstyle\sum_i \lambda_i  \| \textstyle\sum_j \lambda_j \mM \vh_j \|_\sF^2$ converges. Consequently, all sums in the following converge. We simply develop 
\begin{eqnarray*}
 \sum_i \lambda_i \|  \big(\textstyle\sum_j \lambda_j  \mM\vh_j\big) -\tfrac{1}{2} \mM\vh_i\|_\sF^2
& =&  \displaystyle\sum_i \lambda_i  \| \textstyle\sum_j \lambda_j \mM \vh_j \|_\sF^2
   - 2\re \ls  \textstyle\sum_j \lambda_j \mM \vh_j , \textstyle\sum_i \lambda_i \tfrac{1}{2} \mM\vh_i\rs_\sF \\
 &  &+ \tfrac{1}{4} \displaystyle\sum_i \lambda_i \|\mM \vh_i \|_\sF^2 \\
 & = &  \tfrac{1}{4} \displaystyle\sum_i \lambda_i  \|\mM \vh_i \|_\sF^2.
\end{eqnarray*}
\end{proof}

\begin{theorem}[RIP condition for stable recovery]\label{th:robust_RIP_eucl_precise}
Assume that $\Sigma$ is a UoS or a cone and define $\delta_{\Sigma}(f) \defin \underset{\vz\in \sT_f(\Sigma) \setminus \{0\}}{\inf} \delta_{\Sigma}(\vz)$, with $\delta_{\Sigma}(\vz) \defin \underset{\vx \in \Sigma}{\sup} \delta_{\Sigma}(\vx,\vz)$, where
\begin{itemize}
\item {\bf UoS setting}
\begin{equation}
\delta_\Sigma(\vx,\vz) = \delta_{\Sigma}^{\textrm{UoS}}(\vx,\vz) \defin  \frac{-\re \ls \vx,\vz \rs}{\|\vx\|_\sH \sqrt{ \|\vx+\vz\|_\Sigma^2 -  \|\vx\|_\sH^2 - 2\re \ls \vx,\vz \rs}}\label{eq:DefDeltaUoSMain}
\end{equation}
\item {\bf Cone setting}
\end{itemize} 
\begin{equation}
\delta_{\Sigma}(\vx,\vz) 
= \delta_{\Sigma}^{\textrm{cone,sharp}}(\vx,\vz)
\defin
\begin{cases}
\displaystyle \delta^{\textrm{UoS}}_{\Sigma}(\vx,\vz)
&\ \mbox{if}\ \re \ls \vx,\vz \rs \geq \|\vx+\vz\|_{\Sigma}^{2}/2; \\
\\
\displaystyle   \frac{-2\re \ls \vx,\vz \rs}{ \|\vx+\vz\|_\Sigma^2- 2\re \ls \vx,\vz \rs}, & \ \mbox{if}\ \re \ls \vx,\vz \rs < \|\vx+\vz\|_{\Sigma}^{2}/2 
\end{cases}\label{eq:DefDeltaConeMain}
\end{equation}
If $\mM$ is continuous on $\sH$ and satisfies the RIP on $\Sigma-\Sigma$ with constant $\delta < \delta_\Sigma(f)$, then for all $\vx_{0}\in\Sigma$, $\|\ve\|_\sF \leq \eta \leq \epsilon$ we have, with $\vx^*$ the result of minimization~\eqref{eq:robust_minimization},
\begin{equation}
\| \vx^*-\vx_{0}\|_\sH \leq\| \vx^*-\vx_{0}\|_\Sigma \leq  \frac{2\sqrt{1+\delta}}{D_{\Sigma}(f,\delta)} \cdot (\eta +\epsilon )
\end{equation}
where 
\begin{equation}
D_{\Sigma}(f,\delta) \defin \inf_{z \in \sT_f(\Sigma) \setminus \{0\}} \sup_{x \in \Sigma }D(x,z,\delta) > 0
\end{equation}
with 
\begin{itemize}
 \item {\bf UoS setting:}
\begin{equation}
 D_{\Sigma}(\vx,\vz,\delta) = D^{\textrm{UoS}}_{\Sigma}(\vx,\vz,\delta)
  \defin \frac{-2 \re \ls \vx,\vz\rs-\delta \|\vx\|_{\sH} \sqrt{\|\vx+\vz\|_{\Sigma}^{2}-\|\vx\|_{\sH}^{2}- 2\re \ls \vx,\vz \rs}}{\|\vx\|_{\sH}+\|\vx+\vz\|_{\Sigma}}. 
\end{equation}
 \item {\bf Cone stetting:} We have $D_{\Sigma}(\vx,\vz,\delta) = D_{\Sigma}^{\textrm{cone,sharp}}(\vx,\vz,\delta)$ with
 \begin{equation}
  D_{\Sigma}^{\textrm{cone,sharp}}(\vx,\vz,\delta) 
 \defin 
 \begin{cases}
 \displaystyle D^{\textrm{UoS}}_{\Sigma}(\vx,\vz,\delta)
 &\ \mbox{if}\ \re \ls \vx,\vz \rs \geq \|\vx+\vz\|_{\Sigma}^{2}/2; \\
 \\
\displaystyle   \frac{-2\re \ls \vx,\vz \rs -\delta \left( \|\vx+\vz\|_\Sigma^2- 2\re \ls \vx,\vz \rs\right)}
{ \|\vx\|_{\sH}+\|\vx+\vz\|_\Sigma}, & \ \mbox{if}\ \re \ls \vx,\vz \rs < \|\vx+\vz\|_{\Sigma}^{2}/2 
 \end{cases}
\end{equation}
\end{itemize}
\end{theorem}
\begin{remark}
This theorem directly implies Theorem~\ref{th:sharp_RIP} and Theorem~\ref{th:robust_RIP_eucl} because   $\delta^{cone} \leq \delta^{cone,sharp}$: when they are different, we have $\delta^{cone,sharp}(x,z)/\delta^{cone}(x,z)= \frac{\sqrt{\|x+z\|_\Sigma^2/\|x\|_\sH^2- 2\re \ls \vx,\vz \rs/\|x\|_\sH^2 -1}}{\|x+z\|_\Sigma^2/\|x\|_\sH^2- 2\re \ls \vx,\vz \rs/\|x\|_\sH^2} \leq 1$. Similarly we have $D_\Sigma^{UoS}(x,z,\delta) \geq D_\Sigma^{cone,sharp}(x,z,\delta)$.
\end{remark} 
The proof of Theorem~\ref{th:robust_RIP_eucl_precise} mainly relies on Lemma~\ref{lem:MainResult} which will soon be stated. This Lemma gives a bound of the norm of elements of  $\sT_f(\Sigma)$. We start by the proof of the theorem then state and prove Lemma~\ref{lem:MainResult}.
\begin{proof}[Proof of Theorem~\ref{th:robust_RIP_eucl_precise}]
Let $\vx_{0} \in \Sigma $ and $\ve \in \sF$ with $\|\ve\|_\sF \leq \eta \leq \epsilon$.  Choose any
\begin{equation}
\vx^* \in \underset{\vx \in \sH}{\tam}\  f(\vx) \; s.t.\;  \|\mM\vx-(\mM\vx_{0}+\ve
)\|_\sF \leq \epsilon
\end{equation} 
Let  $\vz = \vx^* -\vx_{0}$. 
By definition, we have  $f(\vx^*)  \leq f(\vx_{0})$ i.e.  $f(\vx_{0}+\vz)  \leq f(\vx_{0})$, hence $\vz \in \sT_f(\Sigma)$.
We use Lemma~\ref{lem:MainResult}:
\begin{equation}
 \|\vx^{*}-\vx_{0}\|_{\Sigma} = \|\vz\|_\Sigma \leq  2\frac{\sqrt{1+\delta}}{D_{\Sigma}(f,\delta)} \|Mz\|_\sF.
\end{equation}
Following the classical proof of stable $\ell^{1}$ recovery \cite{Candes_2008}, observe that (with the triangle inequality): 
\begin{equation}\label{eq:MainThStep0}
 \|\mM\vz\|_\sF \leq \eta + \epsilon.
\end{equation}
We conclude  
\begin{equation}
  \|\vx^{*}-\vx_{0}\|_{\sH} \leq  \|\vx^{*}-\vx_{0}\|_{\Sigma} \leq 2\frac{\sqrt{1+\delta}}{D_{\Sigma}(f,\delta)} (\eta + \epsilon).
\end{equation}
\end{proof}
\begin{lemma}\label{lem:MainResult}
Under the hypotheses of Theorem~\ref{th:robust_RIP_eucl_precise},  for any $\vz \in \sT_f(\Sigma)$, 
we have 
\begin{equation}
 \|\vz\|_\Sigma \leq  2\frac{\sqrt{1+\delta}}{D_{\Sigma}(f,\delta)} \|Mz\|_\sF.
\end{equation}
\end{lemma}

\begin{proof}
Let $\vz \in \sT_f(\Sigma)$. If $\vz = 0$ the result is trivial. From now on we assume $\vz \neq 0$. Consider an arbitrary $\vx \in \Sigma$ and $\rho \defin \re \ls \vx,\vx+\vz \rs/\|\vx\|_\sH^2$, $\alpha \defin \|\vx+z\|_\Sigma^2/\|\vx\|_\sH^2$. 
To exploit Fact~\ref{fact:norm_formulation}, we consider $\gamma' > 0$ and $\vz'$ such that $(\vx+\vz') \in \bRp \cdot\tconv(\Sigma\cap S(1))$ and $\|\vx+\vz-(\vx+\vz')\|_{\sH} = \|\vz-\vz'\|_{\sH} \leq \gamma'$, and $\lambda_i \in \bRp$, $\vu_i \in \Sigma$,  such that $\sum \lambda_i =1$ and $\vx+\vz' = \sum \lambda_i \vu_i$. 

Consider an arbitrary $\beta \in \bR$ and $\vh_i \defin \vu_i-(1+\beta )\vx$. Remark that $\sum_j \lambda_j  \vh_j -\frac{1}{2} \vh_i= \frac{1}{2}(\left(1 - \beta \right)\vx-  \vu_i) +\vz'$. We apply Fact~\ref{prop:ident0} to obtain the identity
 \begin{equation}\label{eq:stab1}
 \begin{split}
  \sum \lambda_i \|\mM\left( \vu_i-(1+\beta) \vx\right)\|_\sF^2&=\sum \lambda_i \| \mM\left(\left(1- \beta \right)\vx- \vu_i +2 \vz' \right)\|_\sF^2 \\
  &= \sum \lambda_i\left\{  \| \mM\left(\left(1- \beta \right)\vx- \vu_i\right)\|^2  + 4 \re \ls  \mM\left( \left(1- \beta \right)\vx- \vu_i\right), \mM\vz' \rs +4\| \mM\vz' \|_\sF^2 \right\}\\
  &= \sum \lambda_i \| \mM\left(\left(1- \beta \right)\vx- \vu_i\right)\|^2   -4\re \ls  \mM\left(\beta \vx+\vz'  \right), \,\mM\vz' \rs +4\| \mM\vz' \|_\sF^2 \\
   &= \sum \lambda_i \| \mM\left(\left(1- \beta \right)\vx- \vu_i\right)\|^2   -4 \beta \re \ls  \mM\vx, \mM\vz' \rs.
  \end{split}
 \end{equation}
Consider the vectors $\vh_i=\vu_i -(1+\beta) \vx$ and $\left(1 - \beta \right)\vx- \vu_i$. When $\Sigma$ is a UoS, these vectors are in $\Sigma-\Sigma$ for any $\beta \in \bR$. This is still the case when $\Sigma$ is a cone, provided that $\beta \in (-1,\ 1)$. In both cases, they are in $\Sigma-\Sigma$ provided that $\beta \in I$ where $I \defin \bR$ when $\Sigma$ is a UoS, and  $I \defin (-1,1) \subset \bR$ when $\Sigma$ is a cone. 

Thus, when $\beta \in I$, we apply the RIP hypothesis on both vectors (lower RIP on the left side, upper RIP on the right side) and~\eqref{eq:stab1} yields
\begin{equation}\label{eq:stab1bis}
 \begin{split}
\left(1 - \delta\right)\sum \lambda_i \|\vu_i-(1+\beta) \vx\|_\sH^2 &\leq    \sum \lambda_i \|\mM \left(\left(1- \beta \right)\vx- \vu_i\right)\|_\sF^2 -4 \beta \re \ls  \mM\vx, \mM\vz' \rs\\
 & \leq   (1+\delta) \sum \lambda_i \|\left(1- \beta \right)\vx- \vu_i\|_\sH^2   -4 \beta\re \ls  \mM \vx, \mM\vz' \rs.
 \end{split}
\end{equation}
We calculate:
\begin{equation}
  \begin{split}
(1+\delta) &\sum \lambda_i \|\left(1 - \beta  \right)\vx- \vu_i\|_\sH^2 -   (1 - \delta) \sum \lambda_i \| \vu_i-(1+\beta) \vx\|_\sH^2  \\
   =&   (1+\delta) \sum \lambda_i \left(\left(1 - \beta  \right)^2 \| \vx\|_\sH^2 +  \| \vu_i\|_\sH^2\right) 
    -  (1 - \delta)  \sum \lambda_i\left( \|  \vu_i\|_\sH^2+(1+\beta)^2 \|  \vx\|_\sH^2\right)   \\
    &+ 2\left(-(1+\delta) \left(1 - \beta  \right) +(1 - \delta)(1+\beta) \right) \cdot \re \ls \vx,\sum \lambda_i \vu_i\rs \\
   =&   \sum \lambda_i   \left(    (1+\delta) \left(1 - \beta  \right)^2 - (1 - \delta)(1+\beta)^2  \right) \| \vx\|_\sH^2 + 2\delta   \sum \lambda_i \| \vu_i\|_\sH^2 
+ 4\left(\beta-\delta\right) \re \ls \vx,\vx+\vz' \rs\\
=&
\left(2\delta \beta^{2} -4\beta + 2 \delta\right) \cdot \| \vx\|_\sH^2 + 2\delta  \cdot \sum \lambda_i \| \vu_i\|_\sH^2 
+ 4\left(\beta-\delta\right) \re \ls \vx,\vx+\vz' \rs.\label{eq:MainThStep1}
\end{split}
\end{equation}
Combining with~\eqref{eq:stab1bis}, since $\re \ls \vx,\vx+\vz' \rs = \re \ls \vx,\vx+\vz \rs + O(\gamma')$ and (by continuity of $\mM$) $\re \ls  \mM \vx, \mM\vz' \rs = \re \ls  \mM \vx, \mM\vz \rs + O(\gamma')$, we obtain using Fact~\ref{fact:norm_formulation} (taking the limit when $\gamma' \to 0_{+}$ of the infimum over $\vz'$ and $\lambda_{i},\vu_{i}$ of the right hand side in~\eqref{eq:MainThStep1}) that
\[
\left(2\delta \beta^{2} -4\beta + 2 \delta\right) \cdot \| \vx\|_\sH^2 + 2\delta  \cdot \|\vx+\vz\|_{\Sigma}^{2} 
+ 4\left(\beta-\delta\right) \re \ls \vx,\vx+\vz \rs
 -4 \beta\re \ls  \mM \vx, \mM\vz \rs \geq 0.
\]
Since $\vx \in \Sigma$, we have  $\|\vx\|_\sH^2 = \|\vx\|_\Sigma^2$.
Dividing the above inequality by $2\|\vx\|_{\sH}^{2}$ we get: 
\[
\delta \beta^2 -2\beta(1-\QOC) +\delta (1+\alpha-2\QOC)
 -
 2\beta \frac{ \re \ls  M \vx , M\vz \rs }{\|\vx\|_\sH^2} \geq 0.
 \]
Denoting $g_1(\beta|\QOC,\alpha,\delta) \defin \delta  \beta^2 -2\beta(1-\QOC) +\delta(1+\alpha-2\QOC)$, the Cauchy-Schwartz inequality yields 
\begin{equation}
\begin{split}
-g_1(\beta|\QOC,\alpha,\delta) 
&\leq\frac{ 2 |\beta  |\|  M \vx \|_\sF \| M\vz \|_\sF}{ \|\vx\|_\sH^2}.\\
\end{split}
\end{equation}
Using the upper RIP we have $\|\mM \vx\|_{\sF} \leq \sqrt{1+\delta} \|\vx\|_{\sH}$ hence using~\eqref{eq:MainThStep0} we obtain for $\beta \neq 0$
\begin{equation}\label{eq:stab2}
\frac{1}{|\beta|} \cdot \left[-g_1(\beta|\QOC,\alpha,\delta) 
\right] \leq 2 \sqrt{1+\delta} \frac{ \| M \vz \|_\sH}{\|\vx\|_\sH}
\end{equation}
The rest of the proof consists first in optimizing the choice of $\beta \in I$, then that of $\vx \in \Sigma$ to best exploit the above inequality.

{\bf Choice of $\beta$.} We now compute the supremum of the left hand side, $g(\beta) \defin \frac{1}{|\beta |} \left[-g_1(\beta|\QOC,\alpha,\delta) 
\right]=  -\delta |\beta| +2\frac{\beta}{|\beta|}(1-\QOC) -\delta(1+\alpha-2\QOC)\frac{1}{|\beta|}$, over admissible $\beta \in I \setminus\{0\}$.
Since $g(-|\beta|) \geq g(|\beta|)$, and since the set $I$ is symmetric around zero, we need only consider the maximization of the function $g$ over positive $\beta$.  

Since the function $g$ is bounded from above, we have $1+ \alpha-2 \rho  \geq 0$. 
Straightforward calculus shows that $g$ is increasing for $0< \beta< \beta_{\textrm{opt}}$ and decreasing for $\beta > \beta_{\textrm{opt}}$ with
$\beta_{\textrm{opt}} \defin \sqrt{1+\alpha-2\QOC}$.
\begin{itemize}
\item {\bf UoS setting.} $I=\bR$.  The supremum is reached at $\beta_{\textrm{opt}}$. Thus 
\[
\begin{split}
G_{\textrm{UoS}}(\QOC,\alpha,\delta) \defin & \sup_{\beta \in I} g(\beta) =g(\beta_{\textrm{opt}})
= -\delta\sqrt{1+\alpha-2\QOC} +2(1-\QOC)-\delta\sqrt{1+\alpha-2\QOC}\\
=&  2(1-\QOC)-2\delta\sqrt{1+\alpha-2\QOC}.
\end{split}
\]
Moreover, $G_{\textrm{UoS}}(\QOC,\alpha,\delta)>0$ if and only if $\delta < \delta_{\textrm{UoS}}(\rho,\alpha) \defin \frac{1-\QOC}{\sqrt{1+\alpha-2\QOC}}$.
\item {\bf Cone setting.} $I = (-1,1)$. We compute $G_{\textrm{cone}}(\QOC,\alpha,\delta) \defin \sup_{\beta \in I} g(\beta)$ by considering two cases:
\begin{itemize}
\item if $\rho \geq \alpha/2$ then $\beta_{\textrm{opt}}\leq 1$ and $g(\beta)$ is again maximized at $\beta=\beta_{\textrm{opt}}$ yielding 
$G_{\textrm{cone}}(\QOC,\alpha,\delta) =  G_{\textrm{UoS}}(\QOC,\alpha,\delta)$, which is positive when $\delta < \delta_{\textrm{UoS}}(\QOC,\alpha)$;
\item if $\rho < \tfrac{\alpha}{2}$, then $g(\beta)$ is maximized at $\beta=1$, yielding
$G_{\textrm{cone}}(\QOC,\alpha,\delta) =g(1)=  2(1-\QOC) -\delta \left(2+\alpha-2\QOC\right)$, which is positive
 if and only if $\delta < \frac{2(1-\QOC)}{2+\alpha-2\QOC}$.
\end{itemize}
Overall, we obtain
\begin{equation}
G_{\textrm{cone}}(\QOC,\alpha,\delta) \defin \sup_{\beta \in I} g(\beta) =
\begin{cases}
2(1-\QOC)-2\delta\sqrt{1+\alpha-2\QOC},\ \mbox{when}\ \QOC \geq \alpha/2;\\
2(1-\QOC)-\delta(2+\alpha-2\QOC), \mbox{when}\ \QOC < \alpha/2.
\end{cases}
\end{equation}
This is positive if and only if $\delta < \delta_{\textrm{cone}}(\QOC,\alpha)$, with
\begin{equation}
\delta_{\textrm{cone}}(\QOC,\alpha) \defin 
\begin{cases}
\frac{(1-\QOC)}{\sqrt{1+\alpha-2\QOC}},\ \mbox{when}\ \QOC \geq \alpha/2;\\
\frac{2(1-\QOC)}{2+\alpha-2\QOC},\mbox{when}\ \QOC < \alpha/2.
\end{cases}
\end{equation}
\end{itemize}
From now on when $\Sigma$ is a UoS (resp. a cone) we use the notation $G_{\Sigma}(\cdot)$ as a shorthand for $G_{\textrm{UoS}}(\cdot)$ (resp. $G_{\textrm{cone}}(\cdot)$), and a similar convention for $\delta_{\Sigma}(\cdot)$.

\paragraph{Conclusion of the proof}

Under the theorem hypothesis, simple algebraic manipulations of the expressions of $\delta_{\Sigma}(\cdot)$ and $G_{\Sigma}(\cdot)$ above show that we have, with the notations of the theorem,
\[
\delta < \delta_{\Sigma}(\vz) = \sup_{\vx \in \Sigma} \delta_{\Sigma}(\QOC(\vx,\vz),\alpha_{\Sigma}(\vx,\vz)).
\] 
Hence there is indeed at least one $\vx \in \Sigma$ such that $\delta < \delta_{\Sigma}(\QOC,\alpha)$ with $\QOC = \QOC(\vx,\vz)$, $\alpha = \alpha_{\Sigma}(\vx,\vz)$, or equivalently so that $G_{\Sigma}(\QOC,\alpha,\delta)>0$. 
For any such $\vx$, with the triangle inequality, we obtain
\begin{equation}
\|\vz\|_{\Sigma} \leq \|\vx+\vz\|_\Sigma +\|\vx\|_\Sigma
          \leq  (1+\sqrt{\alpha)}\|\vx\|_\Sigma =   (1+\sqrt{\alpha)}\|\vx\|_\sH.          
\end{equation}
Combined with~\eqref{eq:stab2}, this yields $\|\vz\|_\Sigma \leq  2\sqrt{1+\delta}\|\mM\vz\|_\sF/D_{\Sigma}(\vx,\vz,\delta)$ 
with 
\begin{equation}
\begin{split}
D_{\Sigma}(\vx,\vz,\delta) &\defin \frac{G_\Sigma(\QOC,\alpha,\delta)} {1+\sqrt{\alpha}}
\end{split}
\end{equation}
Maximizing $D_{\Sigma}(\vx,\vz,\delta)$ over $\vx$ yields the best stability constant for the considered descent vector $\vz$
\begin{equation}
D_{\Sigma}(\vz,\delta) \defin  \sup_{\vx \in \Sigma} \frac{G_{\Sigma}(\QOC,\alpha,\delta)}{1+\sqrt{\alpha}}. 
\end{equation}
This gives 
\begin{equation}
\|\vx^{*}-\vx_{0}\|_{\Sigma} = \|\vz\|_\Sigma \leq  2\frac{\sqrt{1+\delta}}{D_{\Sigma}(f,\delta)}\|\mM\vz\|_\sF
\end{equation}
with
\begin{equation} \label{eq:D_const_definition}
D_{\Sigma}(f,\delta) \defin  \inf_{\vz \in \sT_{f}(\Sigma) \setminus \{0\} } \sup_{\vx \in \Sigma} \frac{G_{\Sigma}(\QOC,\alpha,\delta)}{1+\sqrt{\alpha}}. 
\end{equation} 
 \end{proof}

\subsection{Robustness to model error - Instance optimality}\label{sec:instance_optimality}

An important question is whether the results of the previous section are robust to model error, i.e., when the observed signal $\vx_{0}$ is not {\em exactly} in $\Sigma$.  For arbitrary $f$, it is not difficult to derive such robust instance optimality results by directly using Theorem~\ref{th:robust_RIP_eucl} and considering modeling error as observation noise. However, in this case, we would need to know the level $d(\vx_{0},\Sigma)$ of modeling error in order to ``tune'' $\epsilon$ in the optimization problem~\eqref{eq:minimization}. 

Instead, we establish below robustness results that do not require such prior knowledge at the price of certain assumptions on the regularizer $f$. We first give an upper bound on the reconstruction error in Lemma~\ref{lem:inst_opt} then show in Theorem~\ref{th:inst_opt2} that this bound implies instance optimality for $f$ with certain properties. We will give examples in classical cases after the statement of our result. 

Lemma~\ref{lem:inst_opt} uses the notion of $M$-norm \cite{Bourrier_2014}: given a constant $C$, the $M$-norm is defined by \begin{equation}
\label{eq:DefMNorm}
\|\cdot\|_{M,C} := C \cdot \|M\cdot\|_\sF + \|\cdot\|_\sH.
\end{equation}
\begin{lemma}\label{lem:inst_opt}
Let $\Sigma$ be a cone or a UoS. Consider a continuous linear operator $\mM$ with RIP $\delta<\delta_\Sigma(f)$ on $\Sigma-\Sigma$, and a noise level $\eta \leq \epsilon$. Let  $C_\Sigma=C_\Sigma(f,\delta) = \frac{2\sqrt{1+\delta}}{D_{\Sigma}(f,\delta)}$.
Then for all $\vx_{0} \in \sH$, $\ve \in \sF$,  such that  $\| \ve \|_\sH \leq \eta \leq \epsilon$, any minimizer $\vx^*$ of~\eqref{eq:robust_minimization} satisfies
\begin{equation}
 \|\vx^*-\vx_{0}\|_\sH \leq  C_\Sigma \cdot (\eta +\epsilon) + \inf_{z' \in \sT_f(\Sigma) } \|\vx^*-\vx_0 -z'\|_{M,C_\Sigma}.
\end{equation}
\end{lemma}
\begin{remark}
 Note that contrary to Theorem~\ref{th:robust_RIP_eucl_precise}, the unknown $x_0$ is no longer restricted to $\Sigma$. The constant  $C_\Sigma=C_\Sigma(f,\delta)$ has the same definition as in Theorem~\ref{th:robust_RIP_eucl_precise}. 
\end{remark}

\begin{proof} 
Let $\vx_{0} \in \sH$, $\ve \in \sF$. Let 
\begin{equation}
\vx^* \in \underset{\vx \in \sH}{\tam}\  f(\vx) \; s.t.\;  \|\mM\vx-(\mM\vx_{0}+\ve
)\|_\sF \leq \epsilon.
\end{equation} 
Let $z = x^*-x_0$. We have $f(x_0+z)= f(x^*) \leq f(x_0)$. Consider any $z' \in \sT_f(\Sigma)$. 
With Lemma~\ref{lem:MainResult}, 
\begin{equation}
   \|z'\|_\Sigma  \leq C_\Sigma \cdot \|Mz'\|_\sF .
\end{equation}
Thus 
\begin{equation}
\|z\|_\sH\leq \|z'\|_\sH + \|z-z'\|_\sH \leq \|z'\|_\Sigma  + \|z-z'\|_\sH \leq C_\Sigma \cdot \|Mz'\|_\sF + \|z-z'\|_\sH \leq  C_\Sigma \cdot \|Mz\|_\sF + \|z-z'\|_\sH + C_\Sigma \|M(z-z')\|_\sF.
\end{equation}
As in Equation~\eqref{eq:MainThStep0},  $\|Mz\|_\sF \leq \eta + \epsilon$ and we obtain
\begin{equation}
\|z\|_\sH \leq  C_\Sigma \cdot (\eta +\epsilon)  +\|z-z'\|_{M,C_\Sigma}.
\end{equation}
Taking the infimum with respect to $z' \in \sT_f(\Sigma)$ yields the result.
\end{proof}

To go further, we  need to replace the bound involving the $M$-norm with an estimate directly measuring a modeling error. For this, given a regularizer $f$ we use the symmetrized ``distance'' of a vector to the model with respect to $f$: 
\begin{equation}
 d_{f}(\vx_0,\Sigma) = \underset{\tilde{\vx} \in \Sigma}{\inf} \frac{f(\vx_0-\tilde{\vx})+f(\tilde{\vx}-\vx_0)}{2}
\end{equation}
Instance optimality is a direct consequence of Lemma~\ref{lem:inst_opt} using an appropriate \emph{robustness constant} .
\begin{definition}
A constant $D \in \bRp \cup \{+\infty\}$ is called  a \emph{robustness constant} of the regularizer $f$ for the recovery of elements of $\Sigma$ from their measurement with $M$ if, under the assumptions of Lemma~\ref{lem:inst_opt}, we have: for all $x_0 \in \sH$,  
\[
\inf_{z' \in \sT_f(\Sigma) } \|\vx^*-\vx_0 -z'\|_{M,C_\Sigma} \leq D \cdot d_f(x_0,\Sigma).
\] 
\end{definition}
 
The following theorem provides a first general route to finding and exploiting such robustness constants assuming that the (positively homogeneous, non-negative and convex) regularizer $f$ dominates the $M$-norm. We will provide more tailored results for structured sparsity in Section~\ref{sec:infinite_dim}.

\begin{theorem}\label{th:inst_opt2}
Let $\Sigma$ be a cone or a UoS. Let $f$ be positively homogeneous, non-negative and convex such that $f(x)<+\infty$ for $x\in \Sigma$. Consider a continuous linear operator $\mM$ with RIP $\delta<\delta_\Sigma(f)$ on $\Sigma-\Sigma$ and a noise level $\eta \leq \epsilon$. Let $ C_{f,M,\Sigma} < \infty$ such that for all $u \in \sH$, $\|u\|_{M,C_\Sigma}\leq  C_{f,M,\Sigma} \cdot f(u)$. Then 
\begin{itemize}
\item $D := 2C_{f,M,\Sigma}$ is a robustness constant as defined above.
\item As a consequence, for all $\vx_{0} \in \sH$, $\ve \in \sF$,  such that  $\| \ve \|_\sH \leq \eta \leq \epsilon$, any minimizer $\vx^*$ of~\eqref{eq:robust_minimization} satisfies
\begin{equation}
 \|\vx^*-\vx_{0}\|_\sH \leq  C_\Sigma \cdot (\eta +\epsilon) + 2 C_{f,M,\Sigma} \cdot d_f(x_0,\Sigma).
\end{equation}
\end{itemize}
\end{theorem}
\begin{proof} 
If $d_{f}(\vx_{0},\Sigma) = \infty$ the result is trivial. We now assume that $d_{f}(\vx_{0},\Sigma) < \infty$.
Let $\vx_{0} \in \sH$, $\ve \in \sF$ and $x^*$ the result of minimization~\eqref{eq:robust_minimization}. Let $z = x^*-x_0$. We have $f(x_0+z) \leq f(x_0)$. 

Consider $\vx_1 \in \Sigma$: with the sub-additivity of $f$ (which follows from convexity),
\begin{equation}\label{eq:Robust1}
f(x_1+z) \leq f(x_0+z) + f(x_1-x_0)\leq f(x_0) + f(x_1-x_0) \leq f(x_1) + f(x_0-x_1) +f(x_1-x_0).
\end{equation}
Consider now $x_2 \in \Sigma$ such that 
\begin{equation}\label{eq:Robust2}
f(x_2+z)\leq f(x_2) + f(x_0-x_1) +f(x_1-x_0).
\end{equation}
For example\footnote{The use of other choices for this intermediate $x_2$ will be useful for the derivation of better robustness constants in particular cases in Section~\ref{sec:infinite_dim}.}
 choose $\vx_2=\vx_1$. 
Since $d_{f}(\vx_{0},\Sigma) < \infty$, both $f(x_0-x_1)$ and $f(x_1-x_0)$ are finite. Let 
\begin{equation}\label{eq:Robust3}
y:= \tfrac{f(x_2)}{f(x_2) + f(x_0-x_1) + f(x_1-x_0)}\cdot (x_2+z);\quad \text{and}\quad z' := y-x_2.
\end{equation}
With the positive homogeneity of $f$, we have $f(x_2+z')=f(y) = \tfrac{f(x_2)}{f(x_2) +  f(x_0-x_1) + f(x_1-x_0)}\cdot f(x_2+z)\leq f(x_2)$. This means that 
\begin{equation}\label{eq:Robust4}
z' \in \sT_f(x_2) \subset \sT_f(\Sigma).
\end{equation}
Moreover remark that
 \begin{equation}\label{eq:Robust5}
 \begin{split}
 z-z'=  x_2 + z -y  &= \left(1 -\frac{f(x_2)}{f(x_2) +  f(x_0-x_1) + f(x_1-x_0)}\right ) \cdot (x_2+z)\\
 &= \frac{  f(x_0-x_1) + f(x_1-x_0)}{f(x_2) +  f(x_0-x_1) + f(x_1-x_0)}  \cdot (x_2+z).\\
 \end{split}
\end{equation}
Thus, with homogeneity
\begin{equation}
f(z-z') =  \frac{  f(x_0-x_1) + f(x_1-x_0)}{f(x_2) +  f(x_0-x_1) + f(x_1-x_0)}  \cdot f(x_2+z) \leq  f(x_0-x_1) + f(x_1-x_0).
\end{equation}
Using the property of $C_{f,M,\Sigma}$ it follows that
\begin{equation}
\inf_{\ti{z} \in \sT_f(\Sigma) } \|\vx^*-\vx_{0}-\ti{z}\|_{M,C_\Sigma} 
\leq 
 \|z-z'\|_{M,C_\Sigma}
 \leq 
 C_{f,M,\Sigma}\cdot f(z-z') \leq 2C_{f,M,\Sigma}\cdot \frac{f(x_0-x_1)+ f(x_1-x_0)}{2}.
\end{equation}
Minimizing over $x_1 \in \Sigma$ yields the first claim. Using Lemma~\ref{lem:inst_opt} yields the second one. 
\end{proof}

We can derive robustness constants for the following cases: 
\begin{itemize}
\item Consider $f$ a convex gauge induced by any bounded closed convex set containing $0$ in a finite-dimensional vector space $\sH$. 

From Table~\ref{tab:norm_prop}, $f$ is positively homogeneous, non-negative, lower semi-continuous, convex and coercive. 
With the positive homogeneity of $f$, we have $\inf_{u \in \sH \setminus \{0\}} \tfrac{f(u)}{\|\cdot\|_{M,C_\Sigma}} = \inf_{v \in S_{\|\cdot\|_{M,C_\Sigma}}}  f(v)$ where $S_{\|\cdot\|_{M,C_\Sigma}}$ is the unit sphere with respect to $\|\cdot\|_{M,C_\Sigma}$.  The sphere $S_{\|\cdot\|_{M,C_\Sigma}}$ is compact because $\sH$ is finite-dimensional and $\|\cdot\|_{M,C_\Sigma}$ is lower semi-continuous. Thus $f$ admits a minimum value $\beta \geq 0$ on $S_{\|\cdot\|_{M,C_\Sigma}}$ because is bounded and closed (because $\|\cdot\|_{M,C_\Sigma}$ is lower semi-continuous). This minimum value is strictly positive because $f$ cannot take the value $0$ outside of the origin (otherwise $f$ being positively homogeneous would not be coercive), and we conclude that $ C_{f,M,\Sigma} = 1/\beta< \infty$ hence $D = 2C_{f,M,\Sigma}$ is a robustness constant and we can apply Theorem~\ref{th:inst_opt2}. 

 \item Consider $f= \|\cdot\|_1$ and $\Sigma = \Sigma_K$  the set of $K$-sparse vectors in a finite-dimensional space $\sH$. Here, $\|\cdot\|_\sH=\|\cdot\|_2$ the usual $\ell^2$-norm. A proof technique for the general case of block structured sparsity, developped in Theorem~\ref{th:inst_opt_block} from the next section, establishes that $D:=2\tfrac{1 + \sqrt{1+\delta} \cdot C_\Sigma}{\sqrt{K}}$ is a robustness constant for the considered problem.
 Lemma~\ref{lem:inst_opt} thus implies 
 \begin{equation} 
 \begin{split}
  \|\vx^*-\vx_{0}\|_2 & \leq C_{\Sigma} \cdot (\epsilon +\eta) +D \cdot d_{\|\cdot\|_1}(\vx_0,\Sigma).
  \end{split}
 \end{equation}
 In the literature, the distance $d_{\|\cdot\|_1}(\vx_0,\Sigma)$  is often written as $\sigma_{K}(\vx_{0})_{1} = \|\vx_{0,T^c}\|_1$, i.e. the best $K$-term approximation for the $\ell^1$-norm.
 The constant $D$ is of the same order $O(1/\sqrt{K})$  found in classical results \cite{Foucart_2013}.

 \item Consider $f= \|\cdot\|_*$ the nuclear norm and $\Sigma = \Sigma_r$ the set of matrices of rank at most $r$. The norm $ \|\cdot\|_\sH =\|\cdot\|_F $ is the Frobenius norm. As explained in Remark~\ref{rem:robustnucnorm} in the next section, $D =  2\tfrac{1 + \sqrt{1+\delta} \cdot C_\Sigma}{\sqrt{r}}$ is a robustness constant for the considered problem. Lemma~\ref{lem:inst_opt} thus implies 
 \begin{equation}
 \begin{split}
  \|\vx^*-\vx_{0}\|_F & \leq C_{\Sigma} \cdot (\epsilon +\eta) +D \cdot d_{\|\cdot\|_*}(\vx_0,\Sigma_r).
  \end{split}
 \end{equation}
 In this case, $d_{\|\cdot\|_*}(\vx_0,\Sigma_r)$ is reached by zeroing the $r$ largest singular values of $\vx_0$.
 \end{itemize}
 
 In the previous examples $\sH$ is a finite-dimensional space. Theorem~\ref{th:inst_opt2} can be used to derive results in some infinite-dimensional setting as well. Such result will be shown in the case of block structured sparsity, in the next section.

\section{Application to block structured sparsity in infinite dimension}\label{sec:structured_sparsity}

Structured sparse models generalize sparse models by including constraints on the (sparse) support of the signal~\cite{Gribonval_2008,Baraniuk_2010,Eldar_2010}. Adding a notion of block sparsity has applications for simultaneous noise and signal sparse modeling and in imaging~\cite{Adcock_2015,Studer_2013, Traonmilin_2015}.  We begin by applying our general results to these models, then we discuss how the resulting framework improves and extends previous known results. 

\subsection{The finite group-sparse model in infinite dimension}

Here, we suppose that $\sH$ is separable. Thus there exists an orthonormal Hilbert basis $(\ve_i)_{i \in \sN}$. Let $G$ be a finite collection of $|G| < +\infty$ non overlapping finite {\em groups}, i.e. supports subsets $g \subset \sN$ with $|g| < \infty$ and $g \cap g' = \emptyset$, $g \neq g'$. The restriction of the vector $\vx \in \sH$ to the group $g$ is $\vx_{g} \defin \sum_{i \in g} \ls \vx,\ve_{i} \rs \ve_{i}$. A group support $H$ is a subset of $G$ and the restriction of $\vx$ to $H$ is $\vx_H \defin \sum_{g \in H} \vx_{g}$. The group support of $\vx \in \sH$, denoted $\supp(\vx)$, is the smallest $H \subset G$ such that $\vx_H=\vx$. 

Given an integer $K$, the $K$-group-sparse model is defined as
\begin{equation}\label{eq:DefKGroupSparseModel}
\Sigma_K \defin \{\vx \in \sH,\ |\supp(\vx)| \leq K\}.
\end{equation}
Considering the atoms 
\begin{equation}\label{eq:DefKGroupSparseAtoms}
\sA \defin \Sigma_1 \cap \sS(1)
\end{equation}
the corresponding  atomic norm is associated to the finite-dimensional space 
\[
\sE(\sA) = \linspan(\{\ve_{i}\}_{i \in \cup_{g \in G} g})
\]
and simply given by
\begin{equation}\label{eq:DefKGroupSparseAtomicNorm}
\|\vx\|_\sA = 
\begin{cases}
\sum_{g\in G} \|\vx_g\|_\sH,\ &\vx \in \sE(\sA);\\
+\infty,\ & \vx \notin \sE(\sA)
\end{cases}
\end{equation}
and its dual norm is
\begin{equation}
\|\vx\|_\sA^{*} = 
\max_{g\in G} \|\vx_g\|_\sH. 
\end{equation}

Theorem~\ref{th:robust_RIP_eucl} can be leveraged to show that $\|\cdot\|_\sA$ yields stable recovery under a RIP hypothesis.
To establish this, we show in the following section that for any $\vz \in \sT_\sA(\Sigma) \setminus \{0\}$ we can find decompositions $\vz = \vx+\vz-\vx$ such that : a) $\QOC(\vx,\vz)=0$ and b) $\alpha_{\Sigma}(\vx,\vz) = 1$.

Because $\|x\|_\sA = +\infty$ if $x \notin \sE(\sA)$, we can do the following proofs exclusively in the finite-dimensional space $\sE(\sA)$  where $\|\cdot\|_\sA$ matches the classical $\ell^1-\ell^2$ norm \cite{Yuan_2006} : then $\|\vx\|_\sA = \sum_{g\in G} \|\vx_g\|_2$. Yet, we can use Theorem~\ref{th:inst_opt2} to obtain robustness results in the whole space $\sH$, by appropriately defining the behaviour of the regularizer on elements of $\sH \setminus \sE(\sA)$ (Section~\ref{sec:infinite_dim}).

\subsection{Decompositions for a ``sharp'' admissible RIP constant}

We will show that the decompositions defined by the following sets $\sY_{\Sigma}(\vz,\sA)$ lead to a bound on $\delta_\Sigma(f)$ that matches the sharp bound of Cai and Zhang and extends it to the considered general setting of structured sparsity. Let
\begin{equation}\label{eq:optimal_decomposition}
 \sY_{\Sigma}(\vz,\sA)\defin \underset{\vx \in \Sigma}{\tam}  \{ \|\vx+\vz\|_\sA-\|\vx\|_\sA\} \subset \Sigma.
\end{equation}
We first characterize $\sY_{\Sigma}(\vz,\sA)$ and show that $x \in \sY_{\Sigma}(\vz,\sA)$ implies $\rho(x,z)=0$.
\begin{proposition}
\label{prop:opt_NCP_group}
 Consider $\vz \in \sE(\sA)$. The set $\sY_{\Sigma}(\vz,\sA)$ is exactly the collection of all vectors $\vx = -\vz_{H}$ where
 $H$ is some group support made of $K$-groups with largest individual Euclidean norms : $|H|=K$ and 
 \[
 \|\vz_{g}\|_{\sH} \geq \|\vz_{g'}\|_{\sH},\ \forall g \in H,\ \forall g' \notin H.
 \]
As a consequence, $\sY_{\Sigma}(\vz,\sA) \neq \emptyset$ and any $\vx \in \sY_{\Sigma}(\vz,\sA)$ satisfies
 \[
 \min_{g \in H} \|\vx_{g}\|_{\sH} \geq \max_{g' \notin H} \|\vz_{g'}\|_{\sH} = \|\vx+\vz\|_\sA^{*}.
 \]
 and $\QOC(\vx,\vz)=0$.
 \end{proposition}
\begin{proof}
 Let $\vz \in \sE(\sA)$.  We minimize the following expression with respect to $\vx \in \Sigma_K$  :
\begin{equation}
\|\vx+\vz\|_\sA-\|\vx\|_\sA= \sum_{g \in G} (\|\vx_g+\vz_g\|_\sH-\|\vx_g\|_\sH)
\end{equation}
For a given group support $H$, we first minimize this expression over the value of $\vx$ under the constraint $\supp(\vx) = H$. Because there is no overlap, we can minimize this sum over each group separately. Denoting $\ti{\vx}_{H}$ a minimizer, each summand is lower bounded (by a triangle inequality) by $-\|\vz_{g}\|_{\sH}$, a value that is reached by setting $\ti{\vx}_{g} \defin -\vz_{g}$. Thus
\begin{equation}
\|\ti{\vx}_{H}+\vz\|_\sA-\|\ti{\vx}_H\|_\sA= \sum_{g \in H^c} \|\vz_g\|_\sH-\sum_{g \in H} \|\vz_g\|_\sH.
\end{equation}
We minimize over $H$ under the constraint $|H| \leq K$ and obtain the result. 
\end{proof}

We now give a bound on $\alpha_{\Sigma}(\vx,\vz)$ for such decompositions.  We need the following lemma which extends the sparse decomposition of polytopes from \cite{Cai_2014} to the general case where $\Sigma_K$ is made of combinations of $K$ pairwise orthogonal elements of $\sA$.

\begin{lemma}
\label{lem:CS_struct_finite}
Let $\sA$ be a set of normalized atoms and $\vu=\sum_{i=1,L} c_i a_i$ with $c_i\in \bR, a_i \in \sA$, such that the $a_i$ are pairwise orthogonal. Let $\Sigma$ be any model set containing all such combinations for $L = K$ (e.g., consider $\Sigma = \Sigma_K(\sA)$). Then 
\[
\|\vu\|_{\Sigma} \leq \tmax ( \sum_{i=1,L} |c_i| / \sqrt{K}, \tmax_{i=1,L} |c_i| \sqrt{K}).
\]
 \end{lemma}
 \begin{proof}
 This proof is a direct extension of the proof of Cai \cite{Cai_2014}. We proceed by induction on $L$. 
 
 For $L \leq K$,  we have $\vu \in \Sigma$ thus 
 \begin{equation}
 \|\vu\|_{\Sigma}^2 = \|\vu\|_\sH^2 = \sum |c_i|^2  \leq  \left( \sum_{i=1,L} |c_i| \right) \left( \max_{i=1,L} |c_i|\right) \leq  \tmax \left( \sum_{i=1,L} |c_i| / \sqrt{K}, \max_{i=1,L} |c_i| \sqrt{K}\right)^{2}.
\end{equation}

Suppose now the statement is true for a given $L \geq K$. Let $\vu=\sum_{i=1}^{L+1} c_i a_i $ with all $c_i \neq 0$ and the $a_i$ pairwise orthogonal. Note $u_i=c_i a_i$. Denoting $\alpha = \tmax (\sum_{i=1}^{L+1}|c_i| / \sqrt{K}, \max_{i=1,L+1} |c_i| \sqrt{K})$, we have $\sum_{i=1}^{L+1} |c_i|\leq \alpha\sqrt{K}$ and $\tmax_{i=1,L+1} |c_i| \sqrt{K}\leq \alpha/\sqrt{K}$. Without loss of generality we (re)order the decomposition such that $|c_1| \geq |c_2| \geq \ldots \geq |c_{L+1}| > 0$. 

For $1 \leq j \leq L+1$ define $\sigma_{j} \defin \sum_{k=j}^{L+1} |c_k|$. By definition, for $j=1$ we have $\sigma_{j} = \sigma_{1} =  \sum_{i=1}^{L+1} |c_i|\leq \alpha \sqrt{K} = \frac{K}{\sqrt{K}}\alpha $. Since $L \geq K$, we obtain $ \sigma_{1}  \leq (L+1-j)\alpha/\sqrt{K}$.
For $j=L+1$ we have $\sigma_{j} = |c_{L+1}| > 0 = (L+1-j)\alpha/\sqrt{K}$. Considering  the largest $1 \leq j^{*} \leq L$ such that $ \sigma_{j}  \leq( L+1-j) \frac{\alpha}{\sqrt{K}}$ for $1 \leq j \leq j^{*}$: we have  
\begin{equation}\label{dem:geom_struct1}
\begin{split}
  \sum_{k=j^{*}}^{L+1} |c_k| \leq( L+1-j^{*}) \frac{\alpha}{\sqrt{K}} \\
  \sum_{k=j^{*}+1}^{L+1} |c_k| > (L-j^{*}) \frac{\alpha}{\sqrt{K}} .
\end{split} 
\end{equation}
Define $\beta \defin \frac{1}{L+1-j^{*}}\sum_{k=j^{*}}^{L+1} |c_k|$. With inequations~\eqref{dem:geom_struct1}, we have $\beta \leq \alpha/\sqrt{K}$ and
\begin{equation}\label{dem:geom_struct2}
\begin{split}
  \max_{j^{*} \leq i \leq L+1} |c_i| = |c_{j^{*}}|& 
  = \sum_{k=j^{*}}^{L+1} |c_k| -\sum_{k=j^{*}+1}^{L+1} |c_k|\\
  &< (L+1-j^*)\beta -(L-j^{*}) \frac{\alpha}{\sqrt{K}} 
  = \beta +  (L-j^{*})(\beta-\alpha/\sqrt{K}) \\
  & \leq \beta \leq  \frac{\alpha}{\sqrt{K}} 
  \end{split}
\end{equation}
For any index $j^{*} \leq i \leq L+1$, defining $\lambda_i \defin \frac{\beta -  |c_i| }{\beta}$ we thus have $\lambda_{i} > 0$, and
\begin{equation}
 \sum_{i=j^{*}}^{L+1} \lambda_i = L+2-j^* -  \frac{\sum_{i=j^{*}}^{L+1}|c_i|}{\beta} = L+2-j^* -  \frac{(L+1-j^*) \beta}{\beta}=1.
\end{equation}
Defining further
\begin{equation}\label{eq:DefVi}
\vv_i \defin \sum_{k=1}^{j^{*}-1} \vu_k + \beta \sum_{k=j^{*}; k \neq i}^{L+1} \frac{\vu_k}{|c_k|},
\end{equation}
 we have
\begin{equation}\label{dem:geom_struct3}
\begin{split}
\sum_{i=j^{*}}^{L+1} \lambda_i \vv_i &= 
\sum_{k=1}^{j^{*}-1} \vu_k
+  
\sum_{i=j^{*}}^{L+1}\left(\beta -  |c_i| \right) \sum_{k=j^{*}; k \neq i}^{L+1} \frac{\vu_k}{|c_k|}\\
  &
  = 
\sum_{k=1}^{j^{*}-1} \vu_k
  +  
 \sum_{i=j^{*}}^{L+1}\left(\beta -  |c_i| \right)   
 \left(-\frac{\vu_i}{|c_i|}
  + 
  \sum_{k=j^{*}}^{L+1} \frac{\vu_k}{|c_k|} \right)\\
  &= 
\sum_{k=1}^{j^{*}-1} \vu_k
  +  
  \sum_{i=j^{*}}^{L+1}\vu_i 
  +
  \sum_{i=j^{*}}^{L+1} \left(-\beta \frac{\vu_i}{|c_i|}+ \left(\beta -  |c_i| \right) \sum_{k=j^{*}}^{L+1} \frac{\vu_k}{|c_k|} \right).\\
\end{split} 
\end{equation}
We calculate the last term. Denoting $w =\sum_{i=j^{*}}^{L+1} \frac{\vu_i}{|c_i|}$, we have
\begin{equation}
\begin{split}
 \sum_{i=j^{*}}^{L+1} \left(-\beta \frac{\vu_i}{|c_i|}+ \left(\beta -  |c_i| \right) \sum_{k=j^{*}}^{L+1} \frac{\vu_k}{|c_k|} \right) =-\beta w + (L+2 -j^*)\beta w - (L+1-j^*)\beta w = 0.\\
 \end{split}
\end{equation}
Along with equation~\eqref{dem:geom_struct3}, this yields
\begin{equation}\label{eq:uconvvi}
 \sum_{i=j^{*}}^{L+1} \lambda_i \vv_i =  u.
\end{equation}
Moreover, 
\begin{equation}
 \sum_{k=1}^{j^{*}-1} |c_k| +\beta  \sum_{k=j^{*}; k \neq i}^{L+1}\frac{|c_k|}{|c_k|} = \sum_{k=1}^{j^{*}-1} |c_k| + \sum_{k=j^{*}}^{L+1} |c_k| \leq \alpha\sqrt{K}
\end{equation}
and
\begin{equation}
\max(\max_{1 \leq k \leq j^{*}-1} |c_k|,\beta) 
 \leq \max(\alpha/\sqrt{K},\beta)
 \stackrel{\eqref{dem:geom_struct2}}{\leq} \alpha/\sqrt{K}
\end{equation}
Finally, we have shown that $\vv_{i} = \sum_{k=1}^L c'_k a_k'$, with pairwise orthogonal $a_k'$ and $\sum_{k=1}^L |c'_k| \leq  \alpha\sqrt{K}$ and $\tmax_{k=1,L} |c'_k| \leq   \alpha/\sqrt{K}$.  We can use the induction hypothesis to get $\|v_i\|_{\Sigma}\leq \alpha$ and use convexity to obtain from~\eqref{eq:uconvvi} :
\begin{equation}
\|\vu\|_{\Sigma}\leq \sum_{i=j^{*}}^{L+1} \lambda_i \|v_i\|_{\Sigma} \leq \sum_{i=j^{*}}^{L+1} \lambda_i \alpha \leq  \alpha.
\end{equation}
\end{proof}

The above lemma has a simple and practical consequence for group-sparsity. 
\begin{corollary}
\label{lem:CS_struct}
Let $\sA$ be the set of normalized 1-group sparse vectors. Then for all $\vu \in \sE(\sA)$ we have 
\begin{equation}
\|\vu\|_{\Sigma_{K}} \leq  \max\left(\frac{\|\vu\|_{\sA}}{\sqrt{K}},\sqrt{K} \|\vu\|_{\sA}^{*}\right).
\end{equation}
 \end{corollary}
 \begin{proof}
 Let $u = \sum_i c_i a_i$  such that $\sum |c_i| = \|u\|_\sA$ and $\ls a_i,a_j\rs=0$. Apply Lemma~\ref{lem:CS_struct_finite}
 \end{proof}

With these results we are now equipped to bound $\alpha(\vx,\vz)$ for $\vx \in \sY_{\Sigma}(\vz,\sA)$ when $\vz$ is a descent vector for the group-norm $\|\cdot\|_{\sA}$ with respect to $\Sigma_{K}$. 

\begin{proposition}
\label{prop:geom_structured}
For any $\vz \in \sT_\sA(\Sigma_{K})\setminus \{0\}, x \in \sY_{\Sigma}(\vz,\sA)$, 
\begin{equation}
\alpha(x,z) \leq 1
\end{equation}
\end{proposition}

\begin{proof}
Let $\vz \in \sT_\sA(\Sigma_{K})\setminus \{0\}$ and $\vx \in \sY_{\Sigma}(\vz,\sA)$. By definition of $\sT_\sA(\Sigma_{K})$ there is $\ti{\vx} \in \Sigma_{K}$ such that $\|\ti{\vx}+\vz\|_{\sA} \leq \|\ti{\vx}\|_{\sA}$, hence by definition of $\sY_{\Sigma}(\vz,\sA)$ we have $\|\vx+\vz\|_{\sA}-\|\vx\|_{\sA} = \min_{\ti{\vx} \in \Sigma_{K}} (\|\ti{\vx}+\vz\|_{\sA} - \|\ti{\vx}\|_{\sA}) \leq 0$.
Let $H$ be the group support of $\vx$. By Proposition~\ref{prop:opt_NCP_group}, 
\[
K\|\vx+\vz\|_\sA^* \leq  K \min_{g \in H} \|\vx_g\|_\sH \leq \sum_{g \in H} \|\vx_{g}\|_{\sH} = \|\vx\|_{\sA}.
\]
Combining the above estimate with $\|\vx+\vz\|_\sA \leq \|\vx\|_\sA$ 
and Lemma~\ref{lem:CS_struct} yields
\[
 \|x+z\|_\Sigma^2  \leq \max(\|\vx+\vz\|_{\sA}/\sqrt{K}, \|\vx+\vz\|_{\sA}^{*} \sqrt{K})^{2} \leq \|x\|_\sA^{2}/K.
\]
Since $\vx$ is $K$-group-sparse, we have
\[
\|\vx\|_{\sA} = \sum_{i=1}^{K} \|\vx_{g_{i}}\|_{\sH} \leq \sqrt{K} \sqrt{\textstyle \sum_{i=1}^{K}\|\vx_{g_{i}}\|_{\sH}^{2}} = \sqrt{K} \|\vx\|_{\sH}
\]
hence $\|x+z\|_\Sigma^2 \leq \|x\|_\sH^2$.
\end{proof}

To summarize, when $\vz \in  \sT_\sA(\Sigma_{K})\setminus \{0\}$ we can exhibit $\vx \in \Sigma_{K}$ such that $\rho(\vx,\vz) = 0$ and $\alpha_{\Sigma_{K}}(\vx,\vz) \leq 1$. It follows that
\begin{equation}
\alpha_{\Sigma_{K}}(f) \defin \sup_{\vz \in  \sT_\sA(\Sigma_{K})\setminus \{0\}} \inf_{\vx \in \Sigma_{K}:\rho(\vx,\vz) = 0} \alpha_{\Sigma_K}(\vx,\vz) \leq 1.
\end{equation}
and we obtain our main theorem for the $K$-group-sparse model.
\begin{theorem}
\label{eq:RIPKGroupSparse}
For the $K$-group-sparse model $\Sigma_{K} \subset \sH$ defined in~\eqref{eq:DefKGroupSparseModel}, and the atomic norm $f(\cdot) = \|\cdot\|_{\sA}$ defined in~\eqref{eq:DefKGroupSparseAtomicNorm}, we have
\[
\delta_{\Sigma_K}(\|\cdot\|_{\sA}) \geq \frac{1}{\sqrt{2}}.
\]
\end{theorem}

\subsection{Extension to block structured sparsity}\label{sec:block_struct_sparse}

We now show that we can extend the above results to \emph{block structured sparsity}. Consider a collection of $J$ finite-dimensional orthogonal spaces $E_j \subset \sH$ each equipped with a $K_j$-group-sparse model $\Sigma_j$ as defined in~\eqref{eq:DefKGroupSparseModel} (each with its set of groups $G_j$). For each $j$, we have the associated structured norm $\|\cdot\|_{\sA_j}$ as defined in~\eqref{eq:DefKGroupSparseAtomicNorm}. Since the subspaces are orthogonal, there is a natural isomorphism between their direct sum and their Cartesian product. It is simpler to work with the latter, and the model 
\begin{equation}\label{eq:DefBlockStructuredSparseModel}
\Sigma \defin \left\{\vx \in \sH, \vx = \sum_{j=1}^{J} \vx_{j}, \vx_{j} \in \Sigma_{j}\right\}
\end{equation}
is identified to the Cartesian product of the models $\Sigma_1 \times \Sigma_2 \times \ldots \times \Sigma_J$.  A natural regularizer for this \emph{block structured sparsity} model is defined as follows in $E_{1} \times \ldots E_{J}$ : 
\begin{equation}\label{eq:DefWeightedBlockStructuredNorm}
f_w: (\vx_1, \ldots \vx_J) \mapsto w_1\|\vx_1\|_{\sA_1}+\ldots+w_J\|\vx_J\|_{\sA_J}
\end{equation}
with weights $w_j>0$. We show that for $w_{opt,i}= 1/\sqrt{K_i}$, the admissible RIP constant $\delta_\Sigma(f_{w_{opt}})$ does not depend on the ratios $K_i/K_j$ contrary to the result for $w_i=1$ and simple sparsity from \cite{Bastounis_2015}. 

First we need to characterize the atomic norm $\|\cdot\|_\Sigma$ for this model.

\begin{lemma}[The atomic norm $\|\cdot\|_\Sigma$ for Cartesian products of models]
Consider $J$ spaces $\sH_j$ and cone models $\Sigma_j \subset \sH_{j}$, and their Cartesian product $\Sigma =\Sigma_1 \times \Sigma_2 \ldots \times \Sigma_{J} \subset \sH_1 \times \sH_2 \ldots \sH_{J} = \sH$. The space $\sH$ is naturally equipped with the inner product: 
\begin{equation}
\ls (\vx_1,\ldots,\vx_{J}), (\vy_1,\ldots,\vy_{J}) \rs  =\sum_{j} \ls \vx_j, \vy_j \rs_{\sH_j}.
\end{equation}

We have the characterization
\begin{equation}
\|(\vx_1,\ldots,\vx_J)\|_\Sigma^2 = \sum_j \|\vx_j\|_{\Sigma_j}^2.
 \end{equation}
\end{lemma}
\begin{proof}
Proving the result for $J=2$ is enough to obtain it by induction for any $J > 2$. We use  Fact~\ref{fact:norm_formulation}, assuming that all convex hulls are closed to keep the proof simple. First we get a lower bound
\begin{eqnarray*}
  \|(x_1,x_2)\|_\Sigma^2  &=& \inf \{ \sum \lambda_i \|(\vu_i,\vv_i)\|_\sH^2:  \lambda_i \in \bRp, \sum \lambda_i =1, 
 (\vu_i,\vv_i) \in \Sigma , (x_1,x_2)=\sum  \lambda_i (\vu_i,\vv_i)\}\\
 &=& \inf \{ \sum \lambda_i \|\vu_i \|_\sH^2 + \sum \lambda_i\|\vv_i\|_\sH^2:  \lambda_i \in \bRp, \sum \lambda_i =1, 
 (\vu_i,\vv_i) \in \Sigma , (x_1,x_2)=\sum  \lambda_i (\vu_i,\vv_i) \}\\
 &\geq&\inf \{ \sum \lambda_i \|\vu_i \|_\sH^2 :  \lambda_i \in \bRp, \sum \lambda_i =1, 
 \vu_i \in \Sigma_1 , x_1=\sum  \lambda_i \vu_i\}\\
 &&+ \inf \{ \sum \lambda_i \|\vv_i \|_\sH^2 :  \lambda_i \in \bRp, \sum \lambda_i =1, 
 \vv_i \in \Sigma_2 , x_2=\sum  \lambda_i \vv_i\}\\
 &=& \|x_1\|_{\Sigma_1}^2 +\|x_2\|_{\Sigma_2}^2.
\end{eqnarray*}
To prove the converse bound, consider $(u_i)_{i \in I} \in \Sigma_1, v_{j \in J} \in \Sigma_2$, $\lambda_i>0, \mu_j>0$ such that $\sum \lambda_i=1$,  $\sum \mu_j=1$. We show that there exist $\gamma_k,u_k'\in\Sigma_1,v_k'\in \Sigma_2$ such that $\sum_{k} \gamma_{k} = 1$ and 
 \begin{equation*}
  \sum \gamma_k \left(\|u_k'\|_\sH^2+ \|v_k'\|_\sH^2 \right) = \sum \lambda_i \|u_i\|_\sH^2 + \sum \mu_j \|v_j\|_\sH^2 .
 \end{equation*}
Without loss of generality, up to renumbering, we can assume that $J \subset I \subset \bN$ (the case $I \subset J$ uses the same proof). Then, with $\gamma_k=\lambda_k,u_k'=u_k$ and $v_k'= \sqrt{\frac{\mu_k}{\lambda_k}} v_k$ if $k \in J$, 0 otherwise. We have 
\begin{equation*}
  \sum \gamma_k \left(\|u_k'\|_\sH^2+ \|v_k'\|_\sH^2 \right) = \sum \lambda_k \|u_k\|_\sH^2+ \sum \lambda_k\frac{\mu_k}{\lambda_k} \|v_k\|_\sH^2 = \sum \lambda_i \|u_i\|_\sH^2 + \sum \mu_j \|v_j\|_\sH^2 .
\end{equation*}
Thus $\|(x_1,x_2)\|_\Sigma^2 \leq  \|x_1\|_{\Sigma_1}^2 +\|x_2\|_{\Sigma_2}^2$.
\end{proof}

For block structured sparsity, we are now equipped to bound $\delta_{\Sigma}(f_w)$.
\begin{theorem}\label{eq:RIPBlockGroupSparse}
Consider the block-group-sparse model $\Sigma =\Sigma_1 \times \ldots \times \Sigma_J$, $J \geq 2$, and the regularizer 
\[
f_w(x_1,..,x_J)=\sum_{j=1}^J w_j \|x_j\|_{\sA_j}.
\]
Denoting $\kappa_w = \left( \frac{\max(w_j\sqrt{K_j})}{\tmin(w_j\sqrt{K_j}) }\right)$, we have
\[
\delta_{\Sigma}(f_w) \geq \frac{1}{\sqrt{2+J \kappa_w^2}}.
\]
 In particular when $w_j = \frac{1}{\sqrt{K_j}}$, we have $\delta_{\Sigma}(f_{w}) \geq \frac{1}{\sqrt{2+J}}$.
\end{theorem}
\begin{proof}
Let $\vz= (\vz_j)_{j=1,J} \in \sT_{f_{w}}(\Sigma)\setminus \{0\}$. For all $j$, Let $\vx_j \in \sY_{\Sigma_j}(z_j,\sA_j)$. Considering $\vx= (\vx_j)_j$, we have $\vz \in \sT_{f_w}(\vx)$ and 
\begin{equation}
\sum_{j=1,J} w_j\|\vx_j+\vz_j\|_{\sA_j} \leq \sum_{j=1,J} w_j\|\vx_j\|_{\sA_j}\leq \sum_{j=1,J} w_j \sqrt{K_j}\|\vx_j\|_{\sH} \leq \tmax(w_j \sqrt{K_j})\sqrt{J} \|\vx\|_\sH.
\end{equation}
Moreover, by Proposition~\ref{prop:opt_NCP_group}, $K_j \|\vx_j+\vz_j\|_{\sA_j}^* \leq \|\vx_{j}\|_{\sA_{j}} \leq \sqrt{K}_{j} \|\vx_{j}\|_{\sH_{j}}$, and by
Corollary~\ref{lem:CS_struct} we have :
\begin{equation}
\begin{split}
\|\vx+\vz\|_\Sigma^2 &=  \sum_{j=1,J}\|\vx_j+\vz_j\|_{\Sigma_j}^2 
   \leq  \sum_{j=1,J}\tmax\left( (\|\vx_j+\vz_j\|_{\sA_j}^2/K_j, K_j (\|\vx_j+\vz_j\|_{\sA_j}^*)^2\right) \\
  & \leq \sum_{j=1,J}  \left(w_j^{2}\|\vx_j+\vz_j\|_{\sA_j}^2 /(w_j^{2}K_j) + K_j (\|\vx_j+\vz_j\|_{\sA_j}^*)^2\right) \\
    & \leq   \left( \sum_{j=1,J} w_j\|\vx_j+\vz_j\|_{\sA_j}/(w_j\sqrt{K_j}) \right)^2  +  \sum_{j=1,J}\|\vx_j\|_{\sH_{j}}^2\\    
  & \leq  \left(J\left( \frac{\max(w_j\sqrt{K_j})}{\tmin(w_j\sqrt{K_j}) }\right)^2+1\right)\|x\|_{\sH}^2.\\
\end{split}
\end{equation}
Thus $\QOC(x,z)=\ls x+z,x \rs = 0$, $\alpha_{\Sigma}(f_{w}) \leq \alpha_\Sigma(\vx,\vz) \leq 1+J\kappa_w^2$,  and $\delta_{\Sigma}(f_{w}) \geq 1/\sqrt{1+\alpha_{\Sigma}(f_{w})} = 1/\sqrt{2+J\kappa_w^{2}}$. 
\end{proof}

Even with adjusted weights, our lower bound $\delta_\Sigma(f_{w_{opt}})$ depends on $J$. In light of Bastounis et al.~\cite{Bastounis_2015}, this dependency could be necessary even for $w=w_{opt}$.

\subsection{Consequences and discussion}

We just calculated an admissible RIP constant for the block structured sparse model~\eqref{eq:DefBlockStructuredSparseModel} and regularizer $f_w$~\eqref{eq:DefWeightedBlockStructuredNorm}. We now specify recovery theorems and compare them to previous state of the art results.

\subsubsection{Dimension reduction and the RIP}

Ayaz et al. \cite{Ayaz_2014} gave a uniform recovery result with the mixed $\ell^1-\ell^2$-norm for structured compressed sensing under a RIP hypothesis. They show that a RIP constant  $\delta < \sqrt{2}-1$ for vectors in $\Sigma_{2K}=\Sigma_K - \Sigma_K$ guarantees the recovery of vectors from $\Sigma_K$. We just showed that the RIP constant of Ayaz et al. can be improved to the sharp $\frac{1}{\sqrt{2}}$. In~\cite{Adcock_2013b}, a model of sparsity in levels is proposed. This is in fact a block sparsity model in $\sE(\sA)$ with classical sparsity in each block, which is covered by the model of Section~\ref{sec:block_struct_sparse}.
 In \cite{Bastounis_2015}, Bastounis et al. show in the case of block sparsity that  $f(\cdot) = \sum \|\cdot\|_{\sA_j} = \|\cdot\|_1$ (i.e., with weights $w_{j}=1$, in this case we write $\kappa_w=\kappa_1$ and $\kappa$ represent the ratio of sparsity between blocks) and RIP (called there RIP in levels) $\delta =1/\sqrt{J(\kappa_1+0.25)^2 +1)}$ on $\Sigma-\Sigma$ guarantees recovery. This constant is improved by our constant $\delta_\Sigma(f_{w})\geq 1/\sqrt{2+J}$ when appropriately weighting the norm of each block. Our result further extends the work of Bastounis et al. to general structured sparsity. The following theorem summarizes our result:

\begin{theorem}\label{th:robust_RIP_block_structured}
 Let $f(\vx_1,...,\vx_J)=\sum_{j=1}^{J}\|\vx_j\|_{\sA_j}/\sqrt{K_j}$. For any $\mM$ that satisfies the RIP on $\Sigma-\Sigma = \Sigma_{2K_1} \times \ldots \times \Sigma_{2K_J}$ with constant  $\delta<\delta_0$ we have: for all $\vx_{0} \in \Sigma$,  $\|\ve\|_{\sF} \leq \eta \leq \epsilon$, with $\vx^*$ the result of minimization~\eqref{eq:robust_minimization},
\begin{equation}
\| \vx^*-\vx_{0}\|_\sH \leq \| \vx^*-\vx_{0}\|_\Sigma \leq   C_{\Sigma}(f,\delta) (\eta +\epsilon )
\end{equation}
where : 
\begin{itemize}
 \item For $J=1$, $\delta_0 = \frac{1}{\sqrt{2}}$ and $C_{\Sigma}(f,\delta) \leq  \frac{2\sqrt{1+\delta}}{1-\delta\sqrt{2}}$.
 \item For $J\geq 2$, $\delta_0 =\sqrt{\frac{1}{2+J}}$ and $C_{\Sigma}(f,\delta) \leq  \frac{(1+\sqrt{1+J})\sqrt{1+\delta}}{1-\delta\sqrt{2+J}}$.
\end{itemize}
\end{theorem}

The comparison of RIP constants is summarized in Figure~\ref{fig:comp_RIP_const}.

\begin{figure}[h]
\centering
\subfloat[$f_w$]{\includegraphics[width=0.45\linewidth]{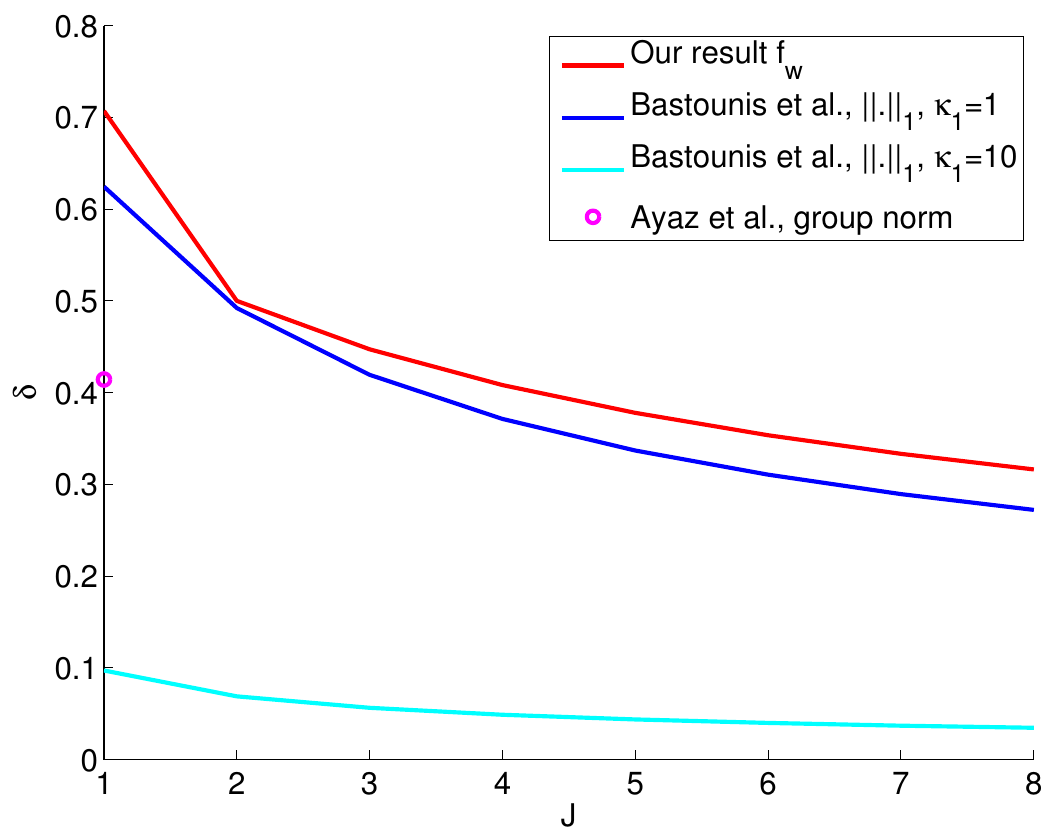}}
\subfloat[$\|\cdot\|_1$]{\includegraphics[width=0.45\linewidth]{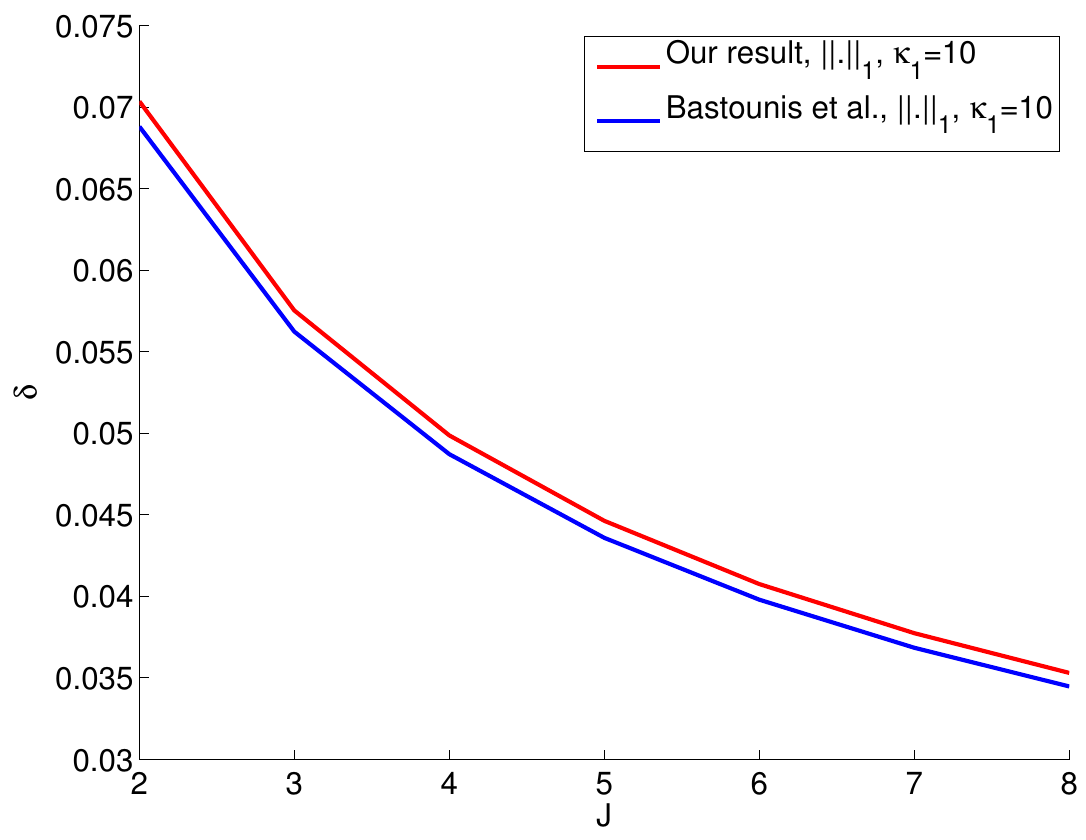}}
\caption{Comparison of our admissible RIP constant and the state of the art (Higher is better). (a) Our result is for weighted group norm with our proposed weights and block group sparsity. Bastounis et al. is for $\ell^1$ norm and simple sparsity. Ayaz et al. is for the group norm. (b) $\kappa_1= 10$. Our result is for the $\ell^1$ norm and block sparsity, Bastounis et al. is for the $\ell^1$ norm and block sparsity.}\label{fig:comp_RIP_const}
\end{figure}

Ayaz et al. \cite{Ayaz_2014} then proceed to show that subgaussian matrices of appropriate dimensions satisfy the RIP on vectors from $\Sigma_K - \Sigma_K$ with high probability, thus providing a sufficient number of observations for guaranteed uniform recovery of structured sparse signals.

Even in an infinite-dimensional Hilbert space, it has recently been established \cite{Puy_2015}  that one can construct random linear measurement operators $\mM: \sH \to \mathbb{R}^{m}$ satisfying the RIP on the secant set $\Sigma-\Sigma$ with high probability, when $m$ is large enough compared to the covering dimension of the normalized secant set $(\Sigma-\Sigma) \cap S(1)$.
The covering number $N(\epsilon)$ of a set is the minimum number of ball of radius $\epsilon$ sufficient to cover the set. If $N(\epsilon) \leq \epsilon^{-s}$ for all $0<\epsilon \leq \epsilon_{0}$, where $\epsilon_{0} \leq 1/2$ and $s>0$, then subgaussian matrices of size $m \times n$ satisfy the RIP with constant $\delta$ on $\Sigma-\Sigma$ with high probability, provided that $m \geq \delta^{-2}O(s \log(1/\epsilon_{0}))$.

For the set $\Sigma_{K}$ of $K$-group sparse vectors, the covering number of $(\Sigma_{K}-\Sigma_{K}) \cap S(1) = \Sigma_{2K} \cap S(1)$ is calculated in \cite{Ayaz_2014} when the groups have fixed size $|g| = r$, for all $g \in G$. 
In the Hilbert space setting considered here, for groups of possibly different sizes, the covering number of $(\Sigma_K-\Sigma_K)\cap S(1)$ can be bounded using the maximum group size  $r \defin \tmax \{ |g| :g \in G \}$. For $\epsilon \leq 1/2$, we obtain $N(\epsilon) \leq \left(\tfrac{3e^{r}(e|G|)}{ K \epsilon} \right)^{K}$. 

Observing that $N(\epsilon) \leq (C/\epsilon)^{s_{0}} = \epsilon^{-2s_{0}} (C \epsilon)^{s_{0}}$ for $\epsilon \leq 1/2$ implies $N(\epsilon) \leq \epsilon^{-2s_{0}}$ for $\epsilon \leq \epsilon_{0} = \min (1/2,1/C)$ gives us immediately the following theorem for structured sparsity with only one block ($J=1$).

\begin{theorem}
For $J=1$, consider $\Sigma = \Sigma_{K}$ the set of $K$-group sparse vectors with groups $G$. Denote $r = \tmax \{ |g| :g \in G \}$.  One can construct a random (subgaussian) linear measurement operator $\mM: \sH \to \mathbb{R}^{m}$ that satisfies the RIP with constant $\delta$ on vectors from $\Sigma_{2K}$ with high probability if :
 \begin{equation}
   m \geq \delta^{-2} O\left(Kr+ K \log(\tfrac{3 e|G|}{K})   \right).
 \end{equation}
\end{theorem}
In the case of block structured sparsity, we extend this to any number of blocks $J$ to determine a sufficient number of (subgaussian) measurements to ensure the RIP on $\Sigma-\Sigma$ holds with high probability. 
 
\begin{theorem}\label{th:subgaussianblocksparse}
For $J\geq 1$, consider $\Sigma = \Sigma_{1} \times \ldots \times \Sigma_{J}$ with $\Sigma_{j} = \Sigma_{K_{j}}$ the set of $K_{j}$ group-sparse vectors with group $G_{j}$. Denote $r_j = \tmax \{ |g| :g \in G_j \}$.  
One can construct a random (subgaussian) linear measurement operator $\mM: \sH \to \mathbb{R}^{m}$ that satisfies the RIP with constant $\delta$ on vectors from 
$\Sigma - \Sigma$ with high probability if :
 \begin{equation}
   m \geq \delta^{-2} 
   O\left(
   \sum_{j=1}^{J} \left(K_{j}r_{j} + K_j \log\left(\tfrac{3e |G_j|}{K_j }\right) \right)\right).
 \end{equation}
\end{theorem}
 \begin{proof}
 We  can bound the covering number of the set $\left( \Sigma - \Sigma \right) \cap S(1)$ by that of $ \left((\Sigma_{1}- \Sigma_{1} ) \cap B(1)\right) \times \ldots \times  \left((\Sigma_J- \Sigma_J ) \cap B(1)\right)$: take $(x_1,..,x_J) \in (\Sigma-\Sigma) \cap S(1)$. Then for all $j$, $\|x_j\|_\sH \leq 1$. Thus $x_j \in (\Sigma_j-\Sigma_j) \cap B(1)$. 
 Note that we use the covering number of $ (\Sigma_j-\Sigma_j) \cap B(1)$ to bound the covering number of $ (\Sigma_j-\Sigma_j) \cap S(1)$. As the covering number of a Cartesian product of sets is the product of their covering numbers \cite{Robinson_2010}, we have  when $\epsilon \leq 1/2$: 
\begin{equation}
\begin{split}
N(\epsilon) &\leq \left(\tfrac{3e^{r_1}(e|G_1|)}{ K_1 \epsilon} \right)^{K_1} \times \ldots \times \left(\tfrac{3e^{r_J}(e|G_J|)}{K_J\epsilon} \right)^{K_J}  \\   
  &=  \left(\frac{1}{\epsilon}\right)^{\sum  K_j} \left( \left( \left(\tfrac{3e^{r_1}(e|G_1|)}{ K_1 }  \right)^{K_1} \times \ldots \times \left(\tfrac{3e^{r_J}(e|G_J|)}{ K_J } \right)^{K_J} \right)^{1/(\sum  K_j)}\right)^{\sum K_j}.  \\ 
\end{split}
\end{equation}
Hence $N(\epsilon) \leq \epsilon^{-s}$ for $0< \epsilon \leq \epsilon_{0}$ with $s = 2\sum K_{j}$ and 
\[
\log 1/\epsilon_{0} = O\left(\frac{\sum_{j} \left(K_{j}r_{j} + \log \left(\tfrac{3|G_j|}{K_j }\right)\right)}{\sum K_{j}}\right).
\]
yields the desired result.
 \end{proof}

Exploiting Theorem~\ref{th:subgaussianblocksparse} in the context of Theorem~\ref{th:robust_RIP_block_structured} considering $\delta < 1/\sqrt{2+J}$ yields a  sufficient number of (subgaussian) measurements to ensure the appropriate RIP for recovery with $f_{w_{opt}}$ and the following number of measurements 
\[
  m \geq 
   O\left(J\sum_{j=1}^{J} \left(K_{j}r_{j} + K_j \log\left(\tfrac{3e |G_j|}{K_j }\right) \right)\right).
\]

The factor $J$ might seem pessimistic, and we attribute its presence to the generality of the result. Should the structure of the observation matrix $M$ be taken into account, better results shall be achieved. In fact, if $M$ is a block diagonal matrix where each block $M_j$ has size  $m_j \times n_j$, uniform recovery guarantees with the $\ell^1$-norm
hold if and only if uniform recovery holds on each block: this is possible as soon as each block $M_j$ of $M$ satisfies the RIP with some constant $\delta_j < \frac{1}{\sqrt{2}}$ on $\Sigma_j-\Sigma_j$, which is in turn exactly equivalent to the RIP with constant $\delta < \frac{1}{\sqrt{2}}$ on $\Sigma-\Sigma$.
The above observation suggests that sharper results for block structured sparsity could be achievable under some structure assumption on $M$. A possible structure assumption could correspond to assuming the
existence of a constant $c$ such that for all$ z = (z_1, .., z_J ) \in \Sigma_1 \times...\times \Sigma_j $:
\begin{equation}
\frac{\|M\vz\|_\sH^2}{\sum \|Mz_j\|_\sH^2} \leq c
\end{equation}

With such a constraint, the dependence of $\delta$ on the number of blocks $J$ may be substantially decreased or
even disappear in Theorem 4.3.

\subsubsection{Robustness to model error and infinite-dimensional context}\label{sec:infinite_dim}

In a series of papers~\cite{Adcock_2012,Adcock_2013,Adcock_2013b,Adcock_2013c,Adcock_2015}, Adcock and Hansen propose a sampling strategy (called generalized sampling) to perform compressed sensing in an infinite-dimensional setting. With our new general framework, we naturally extend the notion of structured sparsity to this infinite-dimensional setting. In fact, we observe that, whereas the ambient space has infinite dimension, the low complexity recovery happens in the finite-dimensional space $\sE(\sA)$. This space $\sE(\sA)$ characterizes the arbitrarily high resolution at which the problem is considered.

The main advantage of this infinite-dimensional setting is the instance optimality result where modeling error is measured in the infinite-dimensional space. 

Note that generally, the $M$-norm needs to be bounded by $f$ to obtain (robust) instance optimality with Lemma~\ref{lem:inst_opt} from Section~\ref{sec:instance_optimality}. 

\begin{theorem}[Instance optimality for block structured sparsity]\label{th:inst_opt_block}
Let us consider the regularizer $f=f_w$ with the adapted weights $w_j=1/\sqrt{K_j}$. Suppose $\mM$ has the RIP with constant $\delta < \delta_\Sigma(f)$ on $\Sigma$. Then for all $\vx_{0}\in \sH$,  $\|\ve\|_\sF \leq \eta \leq \epsilon$ and $\vx^*$ the result of minimization~\eqref{eq:robust_minimization}, we have
\begin{equation}
\| \vx^*-\vx_{0}\|_\sH \leq   C_{\Sigma}(f,\delta)(\epsilon +\eta) + D \cdot d_{f}(\vx_0,\Sigma).
\end{equation}
 We have
 \begin{itemize}
 \item for $J=1$
 \begin{equation}\label{eq:RobConstL1}
D \cdot d_{f}(\vx_{0},\Sigma_{K}) := 2(1 + \sqrt{1+\delta}C_\Sigma) \cdot \frac{d_{\|\cdot\|_{\sA}}(\vx_{0},\Sigma_{K})}{\sqrt{K}}
 \end{equation}
 \item for $J \geq 2$
\begin{equation}
D  := 2\sqrt{2}(1 + \sqrt{1+\delta}C_\Sigma).
 \end{equation}
 \end{itemize}
\end{theorem}
\begin{proof}
 We use a similar argument as in Theorem~\ref{th:inst_opt2}. Let $z=x^*-x_0$. Take $x_1=\vx_{0,T}$ where $T$ is the support of the $K$ greatest groups of $x_0$ in each block, and let $x_2 = -z_{T'}$ where $T'$ is the support of the $K$ greatest groups of $z$ in each block. 
The vector $-\vx_{1}$ is made of the best $K_{j}$-group approximation to $\vx_{0}$ in each block, while the vector $-\vx_{2}$ is similarly the best approximation to $\vz$. As a result $f(x_2+z)-f(x_2)\leq f(x_1+z)-f(x_1)$.  As in the proof of Theorem~\ref{th:inst_opt2} it follows that 
\begin{equation}
\label{eq:Ineq1}
f(x_2+z) \leq f(x_2)+ f(x_0-x_1)+f(x_1-x_0) = f(x_2)+ 2f(x_0-x_1).
\end{equation}
This is classically known as the cone constraint in the case of classical sparsity, \cite{Candes_2008}. 

Given the above properties, we use the definition~\eqref{eq:Robust3} of $z'$ and obtain as in the proof of Theorem~\ref{th:inst_opt2} that $z' \in \sT_f(\Sigma)$ and $z-z'= \frac{  2 f(x_0-x_1)}{f(x_2) +  2 f(x_0-x_1)}  \cdot (x_2+z)$. 
 
 The rest of the proof will now deviate from that of Theorem~\ref{th:inst_opt2}.  With the fact that for any $u \in \sH$,  $\|\mM u\|_2 \leq \sqrt{1+\delta} \|u\|_{\Sigma}$ (using the RIP of $M$, see e.g. \cite[Section IV, eq. (50)]{Bourrier_2014}), we have   
\begin{equation}
\|z-z'\|_{M,C_\Sigma} = \|z-z'\|_{\sH}+C_\Sigma\|M(z-z')\|_{\sF}\leq (1+ \sqrt{1+\delta}C_\Sigma)\|z-z'\|_{\Sigma}.
\end{equation}
We will conclude by establishing the bound $\|z-z'\|_{\Sigma} \leq 2\sqrt{2} f(x_0-x_1)$.
 
 With the definition of $x_2$, by Proposition~\ref{prop:opt_NCP_group} we have  $K_{j} \|(x_2+z)_j\|_{\sA_j}^* \leq \|(x_2)_{j}\|_{\sA_j}$. 
Using the result from Section~\ref{sec:structured_sparsity} that $\|\cdot\|_\Sigma^2\leq \sum_{j=1,J} \max(\|\cdot\|_{\sA_j}/\sqrt{K_j}, \sqrt{K_j} \|\cdot\|_{\sA_j}^*)^2$ (Corrolary~\ref{lem:CS_struct}) and the expression of $z-z'$:
 \begin{equation}
 \begin{split}
 \|z-z'\|_\Sigma^2 &\leq \frac{4 f(x_0-x_1)^2 }{\left(f(x_2) +  2 f(x_0-x_1)\right)^2}\sum_{j=1,J}  \max(\|(x_2+z)_{j}\|_{\sA_j}/\sqrt{K_j}, \sqrt{K_j}  \|(x_2+z)_{j}\|_{\sA_j}^*)^2\\
 &\leq \frac{4 f(x_0-x_1)^2 }{\left(f(x_2) +  2 f(x_0-x_1)\right)^2} \left(\sum_{j=1,J} \|(x_2+z)_{j}\|_{\sA_j}^2/K_j+ \|(x_2)_{j}\|_{\sA_j}^2/K_j\right)\\
 &\leq \frac{4 f(x_0-x_1)^2 }{\left(f(x_2) +  2 f(x_0-x_1)\right)^2} \left(\left(\sum_{j=1,J} \|x_2+z\|_{\sA_j}/\sqrt{K_j}\right)^2+\left(\sum_{j=1,J} \|x_2\|_{\sA_j}/\sqrt{K_J}\right)^2\right)\\
 &=\frac{4 f(x_0-x_1)^2 }{\left(f(x_2) +  2 f(x_0-x_1)\right)^2} \cdot \left(f(x_2+z)^2+f(x_2)^2\right)\\
 &\stackrel{\eqref{eq:Ineq1}}{\leq} 8 f(x_0-x_1)^2 \\
 \end{split}
 \end{equation}
 When $J = 1$ we let the reader check that the constant $8$ above can be replaced with $4$.
\end{proof}

\begin{remark}\label{rem:robustnucnorm}
The above proof extends to rank-$r$ matrices by considering proper adaptations of the choice of $\vx_{1}$ and $\vx_{2}$ and Proposition~\ref{prop:nuc_decomp} in the next section. With $f$ the nuclear norm, this yields the robustness constant 
 \begin{equation}\label{eq:RobConstNuc}
D := 2(1 + \sqrt{1+\delta}C_\Sigma)/\sqrt{r}.
 \end{equation}
\end{remark}

Beyond these examples, to make sense of Theorem~\ref{th:inst_opt2} in the infinite-dimensional setting, the domain where $f$ is finite must be extended outside of $\sE(\Sigma)$ while keeping a finite robustness constant $D$. This can be done on a case-by-case basis when properties of $M$ and $f$ allow to conclude. 
 
For example, as Adcock and Hansen in \cite{Adcock_2013c}, consider the following setting: $\sH=\ell^2(\sN)$ with Hilbert basis $(e_i)_{i=1,+\infty}$. Consider $\Sigma$ a block sparsity model in $(e_1,..,e_N)$. Let $f=\|\cdot\|_1$. Then $f$ is an extension of the definition of $f_{w}$ in $\sE(\Sigma)$ to the whole space $\sH$ (with $w_j=1$ for all $j$). In \cite{Adcock_2013c}, $M$ is a collection of (Fourier) measurements that have \emph{a strong balancing property}. The important fact here is that this property requires $\|M^HM\|_{\infty} \leq C'$ where $\|\cdot\|_{\infty}$ is the maximum of the $\ell^\infty$-norms of the coefficients of $M^HM$ (where $M^H$ is the Hermitian conjugate of $M$). With such an hypothesis, for any $u \in \sH$, we have:  $\|Mu\|_2^2 = | \ls u, M^HMu \rs| \leq \|M^H M u\|_{\infty} \|u\|_1\leq \|M^HM\|_{\infty} \|u\|_1 \|u\|_1 \leq C' \|u\|_1^2 $. Thus in this case the $M$-norm is bounded by the $\ell^1$-norm : $\|\cdot\|_{M,C} \leq (1+C \sqrt{C'})\|u\|_1$.

\section{Worked examples }  \label{sec:other_ex}

As pointed out in Section~\ref{sec:structured_sparsity}, the sharp RIP conditions $\delta \leq 1/\sqrt{2}$ for classical sparse recovery with the $\ell^{1}$ norm are covered by Theorem~\ref{th:robust_RIP_eucl}. We now demonstrate the flexibility of the general framework of this paper by applying it to other examples. 

\subsection{Stable recovery from linear subspaces}

Let $\sH$ be a Hilbert space and $\Sigma \subset \sH$ be a linear subspace of $\sH$.  With such a model, it is easy to characterize the atomic norm $\|\cdot\|_{\Sigma}$: we have $\|\vu\|_\Sigma= \|\vu\|_\sH$ if  $\vu \in \sE(\Sigma)=\Sigma$, $\|\vu\|_\Sigma= +\infty$ otherwise. As a result we can characterize $\delta_{\Sigma}(f)$ for two natural regularizers.

\paragraph{Regularizing with the indicator function} Let $f$ be the indicator function of $\Sigma$: $f(\vu)=0$ if $\vu \in \Sigma$, $f(\vu)= +\infty$ otherwise. One can easily check that the set of descent vectors is $\sT_f(\Sigma) = \Sigma-\Sigma = \Sigma + \Sigma = \Sigma$, hence for any $\vz \in \sT_f(\Sigma)\setminus \{0\}$, we can take $\vx = -\vz \in \Sigma$ which yields $\alpha(\vx,\vz)=0$ and $\rho(\vx,\vz)=0$. We conclude that $\delta_\Sigma(f)=1$.

\paragraph{Regularizing with the atomic norm $\|\cdot\|_\Sigma$} Consider $\vz \in \sT_{\Sigma}(\Sigma)$: by definition, there is $\vx \in \Sigma$ such that $\|\vx+\vz\|_\Sigma \leq \|\vx\|_\Sigma = \|\vx\|_{\sH}$. Since $\|\vu\|_{\Sigma} = +\infty$ when $\vu \notin \Sigma$, this implies $x+z \in \Sigma$ and $z \in \Sigma- \Sigma$. Thus, $ \sT_{\Sigma}(\Sigma) \subset \Sigma - \Sigma = \Sigma$ and, as above, we conclude that $\delta_\Sigma(\|\cdot\|_\Sigma)=1$.

\subsection{Low rank matrix recovery} Let $\sH$ be the set of $m \times n$ (real or complex) matrices equipped with the Frobenius norm and the associated inner product $\langle \vx,\vy\rangle = \text{trace}(\vx \vy^{T})$ (respectively $\re(\text{trace}(\vx\vy^{H})$ for the complex case), $\Sigma =\Sigma_r$ be the model set of $m \times n$ matrices of rank at most $r$, and $\sA= \Sigma_1\cap S(1)$ the set of normalized rank one-matrices. We consider the nuclear norm $f(\cdot) =\|\cdot\|_\sA =\|\cdot\|_*$ as the regularizer and verify that the low-rank matrix recovery result of \cite{Cai_2014} is covered by the general framework of this paper. 

We decompose any matrix $\vz$ as $\vz = (\vz-\vz_r)+\vz_r$ where $\vz_r$ is obtained by truncating the singular value decomposition (SVD) of $\vz$ to its $r$ largest singular values ($\vz_{r}$ is not unique if some singular values are equal). We first show that we have $\|\vz-\vz_r\|_* -\|\vz_r\|_* \leq 0$ as soon as $\vz$ is a descent matrix, $\vz \in \sT_f(\Sigma)$. We use this property to prove that $\delta_\Sigma(f)\geq 1/\sqrt{2}$.
\begin{proposition}
\label{prop:nuc_decomp}
 Consider a matrix $\vz \in \sT_f(\Sigma)$ and its SVD $\vz = USV$. Let  $z_{r} = US_{r}V$ where $S_{r}$ is a restriction of the matrix $S$ to its $r$ largest diagonal entries (usually the $r$ first if there is no tie). Then  $\|\vz-\vz_r\|_* -\|\vz_r\|_* \leq 0$.
 \end{proposition}
 \begin{proof}
By definition, since $\vz \in \sT_{f}(\Sigma)\setminus \{0\}$ there exists $\vx \in \Sigma$ such that $\| \vx+\vz\|_* -\|\vx\|_*\leq0$. 
%
%
For any $m \times n$ matrix $A$, denote $\sigma_i(A)$ the $i$-th singular value ($\sigma_1(A)\geq \ldots \geq \sigma_n(A) \geq 0$). From \cite[Chapter 7]{Horn_1990}, for any matrices $A,B$, the singular values have the following property: $\|A-B\|_{*} = \sum_{i=1}^n \sigma_i(A-B) \geq\sum_{i=1}^n |\sigma_i(A)-\sigma_i(B)|$. Taking $A= \vz$  and $B=-\vx$ and observing that $\sigma_{i}(\vx) = \sigma_{i}(-\vx) = 0$ for $i>r$ yields
\[
 \begin{split}
  \|\vx+\vz\|_* -\|\vx\|_*  &= \sum_{i=1,n} \sigma_i(\vx+\vz) - \sum_{i=1,n} \sigma_i(\vx)\\
&   \geq \sum_{i=1,n} |\sigma_i(\vz) - \sigma_i(-\vx)| - \sum_{i=1,n} \sigma_i(\vx)\\
&   = \sum_{i=1,r} |\sigma_i(\vz) - \sigma_i(-\vx)| - \sum_{i=1,r} \sigma_i(\vx) + \sum_{i=r+1,n} \sigma_i(\vz)\\
&   \geq -\sum_{i=1,r} \sigma_i(\vz) + \sum_{i=r+1,n} \sigma_i(\vz) = \|\vz-\vz_r\|_* -\|\vz_r\|_*.
 \end{split}
\]
\end{proof}

\begin{proposition}\label{prop:nuc_geom}
With $\Sigma = \Sigma_{r}$ the set of rank $r$ matrices and $\|\cdot\|_{*}$ the nuclear norm, we have  $\delta_\Sigma(\|\cdot\|_*) \geq 1/\sqrt{2}$. 
\end{proposition}
\begin{proof} 
Let $\vz \in \sT_{\|\cdot\|_*}(\Sigma_{r})$ a descent matrix for the nuclear norm, and $\vx = -\vz_r$. With Proposition~\ref{prop:nuc_decomp}, we have $\|\vz-\vz_r\|_* -\|\vz_r\|_* \leq 0$. Since $\Sigma$ is a UoS we use~\eqref{eq:DefDeltaUoSMain} to compute
\begin{equation}
 \delta_{\Sigma}(\vx,\vz) = \frac{\|\vx\|_2^2}{\|\vx\|_2 \sqrt{ \|\vx+\vz\|_\Sigma^2 -\|\vx\|_2^2 +2\|\vx\|_2^2  }}= \frac{1}{ \sqrt{ \|\vz-\vz_r\|_\Sigma^2/\|\vz_r\|_2^2 +1 }}
\end{equation}
Recalling that the operator norm of a matrix $A$ is $\|A\|_{op}= \tmax_i \sigma_i(A) $ and that the SVD gives an optimal atomic decomposition for the nuclear norm with pairwise orthogonal atoms, Lemma~\ref{lem:CS_struct_finite} yields
\begin{equation}
\|\vz-\vz_{r}\|_{\Sigma}
 \leq \max(\|\vz-\vz_r\|_*/\sqrt{r}, \|\vz-\vz_r\|_{op}\sqrt{r} ).
\end{equation}
We have $r\|\vz-\vz_r\|_{op} \leq  \|\vz_r\|_*$ (similarly to Proposition~\ref{prop:opt_NCP_group}) and $\|\vz-\vz_r\|_* \leq \|\vz_r\|_*$, hence $\|\vz-\vz_{r}\|_{\Sigma} \leq  \|\vz_r\|_* /\sqrt{r} \leq \|\vz_r\|_2$ and  we conclude that $\delta_{\Sigma}(\vx,\vz) \geq 1/\sqrt{2}$. Consequently, $\delta_\Sigma(\|\cdot\|_*)\geq 1/\sqrt{2}$
\end{proof}

With Proposition~\ref{prop:nuc_geom}, we have verified: when $\mM$ satisfies the RIP with constant $\delta < \frac{1}{\sqrt{2}}$ on $\Sigma_r-\Sigma_r=\Sigma_{2r}$  (the set of matrices with rank at most $2r$), nuclear norm regularization is guaranteed to yield stable uniform recovery of matrices in $\Sigma_{r}$.

\section{Regularizing with the model's atomic norm $\|\cdot\|_{\Sigma}$ ?}\label{sec:sigma_norm}

The atomic norm $\|\cdot\|_\Sigma$ intrinsically associated to the considered model set $\Sigma$ appears in the definition of the RIP constant $\delta_{\Sigma}(f)$. Can this norm be directly used as a regularizer $f$, with recovery guarantees ? 

In simple cases such as $k$-sparse vector recovery with $k=1$, we have seen that regularizing with $f(\cdot)  = \|\cdot\|_{\Sigma}$ indeed yields uniform recovery. We show below that this ``natural'' regularizer can still used for arbitrary cone models somehow close to the set of one-sparse vectors: unions of one-dimensional half-lines. We then show examples where it is impossible to achieve uniform recovery with $\|\cdot\|_\Sigma$ .

\subsection{Recovery of a finite union of one-dimensional half-spaces} \label{sec:one_dim_union}

Consider $\Sigma$ a cone such that $\Sigma \cap S(1)$ finite: $\Sigma$ is a finite union of one-dimensional half-spaces (half-{\em lines}). A particular example is the set of $k$-sparse vectors with $k=1$. In such a simple setting, for any $\vz \in \sE(\Sigma) $ there is an optimum atomic decomposition $c_i \in \bRp, \va_i \in \Sigma\cap S(1)$ such that $\sum c_i \va_i = \vz$ and $\|\vz\|_\Sigma = \sum c_i$. Because $\Sigma \cap S(1)$ is finite, the coherence 
\begin{equation}
\label{eq:DefCoherence}
 \mu(\Sigma) \defin \max_{\vx, \vy \in \Sigma \cap S(1), \vx \neq \pm \vy} |\ls \vx,\vy \rs|
 \end{equation}
 satisfies $\mu(\Sigma)<1$. We use it to show that uniform recovery is possible by regularization with $\|\cdot\|_\Sigma$.

\begin{proposition}\label{prop:union_1d_halfspace}
If $\Sigma$ is a cone such that $\Sigma \cap S(1)$ is a finite subset of $S(1)$, then
\begin{equation}
\label{eq:RIPCoher}
\delta_\Sigma(\|\cdot\|_\Sigma)\geq \frac{2(1-\mu(\Sigma)) }{ 3+2\mu(\Sigma) }>0.
\end{equation}
If additionally $\Sigma$ is a UoS we have 
\begin{equation}
 \delta_\Sigma(\|\cdot\|_\Sigma)\geq \frac{1-\mu(\Sigma) }{ \sqrt{2(1+\mu(\Sigma)) }}>0.
\end{equation}
\end{proposition}

\begin{proof}
Consider a descent vector $\vz \in \sT_\Sigma(\Sigma)\setminus \{0\}$. By definition there is $\vx \in \Sigma$ such that $\|\vx+\vz\|_\Sigma \leq \|\vx\|_\Sigma = \|\vx\|_{\sH}$, and we can decompose $\vx+\vz$ as a positive linear combination of $\va_i$ (otherwise we would have $\|\vx+\vz\|_\Sigma = +\infty$). Let $\vx+\vz= \sum c_i \va_i$ be an optimal atomic decomposition of $\vx+\vz$ such that $\|\vx+\vz\|_\Sigma= \sum_i c_i$, $c_i \geq 0$ and $ a_{i_1} \neq \pm a_{i_2}$ for any $i_1, i_2$. 
Let $i_{0}$ such that $\vx= \|\vx\|_{\sH} \va_{i_0}$ and $\vx' \defin \vx- c_{i_{0}} \va_{i_0} =(\|\vx\|_{\sH}-c_{i_{0}})\va_{i_0}$. Since $\vx'+\vz =  \sum_{i\neq {i_0}} c_i \va_{i}$ we have
\begin{equation}
\|\vx'+\vz\|_\Sigma \leq \sum_{i\neq {i_0}} c_i  \leq \sum_{i} c_i -c_{i_0} = \|\vx+\vz\|_{\Sigma}-c_{i_{0}} \leq \|\vx\|_{\sH} -c_{i_0}. 
\end{equation}
Thus $\|\vx\|_{\sH} \geq c_{i_0}$ and $\vx' \in \Sigma$ with $\|\vx\|_{\sH}-c_{i_0} = \|x'\|_\sH \geq \|\vx'+\vz\|_{\Sigma}$.
Moreover
\begin{equation}
\begin{split}
|\ls \vx' , \vx'+\vz \rs |&=| \ls \vx', \sum_{i\neq {i_0}}  c_i a_i\rs|
	        \leq   \sum_{i \neq {i_0}} c_i |\ls x',a_i\rs |
	        \leq (\sum_{i \neq i_{0}} c_{i}) \mu \|\vx'\|_{\sH}
	        \leq \mu \|\vx'\|_\sH^2.
\end{split}
\end{equation}
We have thus found $\vx' \in \Sigma$ such that $ (1+\mu)\|\vx'\|_\sH^2 \geq -\ls \vx',\vz\rs \geq (1-\mu) \|\vx'\|_\sH^2$  and, with the expression of $\delta_\Sigma(\|\cdot\|_\Sigma)$ for cones (Equation~\eqref{eq:DefDeltaCone}):
\begin{equation}
\delta_\Sigma(\|\cdot\|_\Sigma) \geq  \frac{2(1-\mu) \|\vx'\|_\sH^2}{ \|\vx'+z\|_\Sigma^2 +2(1+\mu)\|\vx'\|_\sH^2 } \geq  \frac{2(1-\mu)}{ 3 +2\mu }.
\end{equation}

If additionally, $\Sigma$ is a UoS, we have (Equation~\eqref{eq:DefDeltaUoS})
\begin{equation}
\delta_\Sigma(\|\cdot\|_\Sigma) \geq  \frac{(1-\mu) \|\vx'\|_\sH^2}{ \|\vx'\|_\sH \sqrt{\|\vx'+z\|_\Sigma^2 -\|\vx'\|_\sH^2  +2(1+\mu)\|\vx'\|_\sH^2} }   \geq  \frac{1-\mu}{ \sqrt{2(1+\mu)} }.
\end{equation}
\end{proof}

When $\mu =0$, $\delta_\Sigma(\|\cdot\|_\Sigma) \geq 2/3$ for cones and $\delta_\Sigma(\|\cdot\|_\Sigma) \geq 1/\sqrt{2}$ for UoS. This last case is exactly the result for 1-sparse vectors recovery and the $\ell^1$-norm ($\|\cdot\|_\Sigma$ is the $\ell^1$-norm).

\subsection{Recovery of a finite point cloud} Let $\Sigma$ be a finite point cloud containing $r$ points in a Hilbert space $\sH$. Consider $\sA =(\bRp\Sigma) \cap S(1)$, the set of normalized elements of $\Sigma$. Then $\|\cdot\|_\sA =\|\cdot\|_{\Sigma'}$ where $\Sigma'=\bRp\Sigma$. We just saw that elements of $\Sigma' \supset \Sigma$ can be stably recovered by regularization with $\|\cdot\|_{\Sigma'}$ provided that the measurement operator $\mM$ satisfies the RIP on the secant set $\Sigma'-\Sigma' = \bRp\Sigma-\bRp\Sigma$ with an appropriate constant $\delta$. From \cite{Puy_2015}, we just need to calculate the covering number of the normalized secant set $(\bRp\Sigma-\bRp\Sigma) \cap S(1)$ to construct random subgaussian measurements satisfying the desired RIP on this set with high probability. 

This normalized secant 
set is the intersection of the unit sphere with a finite union of $\binom{r}{2}$ two-dimensional subspaces. Thus for $0<\epsilon \leq 1/2$,  $(\bRp\Sigma-\bRp\Sigma) \cap S(1)$ is covered by $N(\epsilon)$ balls with :
\[
\begin{split}
N(\epsilon) &\leq  \binom{r}{2}\left(\frac{C}{\epsilon} \right)^2 
            \leq  \left( \frac{r e}{2} \right)^2 \left(\frac{C}{\epsilon} \right)^2.
\end{split}
\]
Consequently, from \cite{Puy_2015}, it is possible to construct a random linear operator $\mM:\sH \to \bR^{m}$ with 
\(
m \geq O(\log(r)/\delta^2)
\) 
satisfying with high probability the RIP with constant $\delta$ on the secant set. This number of measurements is of the same order as in the non-uniform guarantees from \cite{Chandrasekaran_2012} provided that the required RIP constant $\delta$ does not depend on the dimensions of the problem (see next section for an example), which is not obvious given the possible dependency on $\mu(\Sigma')$. In the Johnson-Lindestrauss Lemma, the existence of measurement operators with RIP on $\Sigma-\Sigma$ (as opposed to $\Sigma'-\Sigma'$) is guaranteed provided $m \geq O(\log(r)/\delta^2)$ \cite{Johnson_1984}. In our case, the same dimension relations are required. We however give a {\em convex} decoder $\|\cdot\|_{\Sigma'}$ that can recover elements of the point cloud. This decoder my not be easily computable if no other hypothesis is available on the model set.

\subsection{Recovery of permutation matrices} Based on the above generic analysis for finite point clouds, we can deduce a number of measurements sufficient for uniform recovery of permutation matrices in $\bR^{n\times n}$ (thus $\sH$ is a real Hilbert space in this case). In this setting, $\Sigma'_{n}=\bRp\cdot \Sigma_{n}$ where the set $\Sigma_{n}$ is made of $r = n!$ points, and we deduce that $m \geq O(n \log(n)/\delta^2)$ measurements are sufficient for uniform recovery. It leaves us however to determine $\delta_{\Sigma'_{n}}(\|\cdot\|_{\Sigma'_{n}})$. 

If we use the generic correlation bound of $\delta_{\Sigma'_{n}}(\|\cdot\|_{\Sigma'_{n}})$ from Proposition~\ref{prop:union_1d_halfspace}, the coherence depends on the dimension of the problem: we let the reader check that $\mu(\Sigma'_{n})= 1- 2/n$ and that $\delta_{\Sigma'_{n}}(\|\cdot\|_{\Sigma'_{n}}) = O(1/n)$, suggesting $m \geq O(n^{3} \log(n))$ which is even larger than the ambient dimension $n^{2}$.  

In fact, we can bound $\delta_{n} = \delta_{\Sigma'_{n}}(\|\cdot\|_{\Sigma'_{n}})$ more precisely by a constant independent of the ambient dimension, showing that stable uniform recovery of permutation matrices with the regularizer $\|\cdot\|_{\Sigma'_{n}}$ is possible with $O(n \log(n))$ measurements. This qualitatively extends to {\em uniform stable recovery} the non-uniform recovery results of \cite{Chandrasekaran_2012}, and also shows that the dimension argument holds beyond Gaussian measurements.

\begin{proposition}
Let $\Sigma$ be the set of $n\times n$ permutation matrices and $\Sigma' = \bR^{+} \Sigma$. We have  
\begin{equation}\label{eq:RIPPermutations}
\delta_{\Sigma'}(\|\cdot\|_{\Sigma'}) \geq \frac{2}{3}.
\end{equation}
\end{proposition}
\begin{proof}
Here $\Sigma'$ is a cone, and $\|.\|_{\Sigma'}$ is the gauge generated by the Birkhoff polytope of bi-stochastic matrices (non-negative matrices which sums of values over each row and column equal 1).  
The norm $\|\vx\|_{\Sigma'}$ is thus infinite unless $\vx = \alpha \vu$ with $\alpha\geq 0$ and $\vu$ is bi-stochastic, in which case $\|\vx\|_{\Sigma'} = \alpha$ is the sum of the entries of $\vx$ over the first row (or any other row or column since all such sums are equal).

Let $\vz \in \sT_{\Sigma'}(\Sigma')\setminus \{0\}$. There is $\vx \in \Sigma'$ such that $\|\vx+\vz\|_{\Sigma'} \leq \|\vx\|_{\Sigma'} = \|\vx\|_{\sH}< \infty$ hence $\vx+\vz = \alpha \vu$ for $0 \leq \alpha \leq \|\vx\|_{\sH}$ and some bi-stochastic matrix $\vu$.
  In particular we must have $\vz(i,j) \geq -\vx(i,j)$ for all row and column indices $i,j$, as well as $\sum_j (\vx(i,j) +\vz(i,j)) = \|\vx+\vz\|_{\Sigma'} \leq \|\vx\|_{\Sigma'} = \sum_j \vx(i,j)$ for all $i$. This implies $\sum_j \vz(i,j) \leq 0$ for all $i$. Thus each line (and similarly each column) of $\vz$ has exactly one negative coordinate and $\vz = \alpha \vu - \beta \vv$  where $\vu$ is bi-stochastic, $\vv$ is a permutation matrix and $\alpha \leq \beta = \|\vx\|_{\Sigma'}$. Reciprocally, all such $\vz$ can be checked to be elements of $\sT_{\Sigma'}(\Sigma')$.  Without loss of generality we can consider the case $\beta =1$, $0 < \alpha \leq 1$ and (up to permuting the lines and columns of all considered matrices) we can assume that $\vv$ is the identity matrix. 

We are now going to choose $\vx_0 = \beta_0 \vv$ to maximize $\delta_{\Sigma'}(\vx_{0},\vz)$. First, we wish to ensure that  $ \|\vx_0+\vz\|_{\Sigma'} < + \infty$. This is equivalent to $\beta_{0}\vv + \alpha \vu-\vv\geq 0$ entry-wise. The only constraint is on the diagonal and reads $\beta_0-1 +\alpha u(i,i) \geq 0 $ for all $i$, which we rewrite $\beta_0\geq 1-\alpha u_{min}$ with $u_{min} \defin \min_{i} u(i,i)$. We then compute using~\eqref{eq:DefDeltaCone}
\[
\begin{split}
\delta_{\Sigma'}(\vx_0,\vz) &= \frac{-2\re \ls \vx_0,\vz \rs}{ \|\vx_0+\vz\|_{\Sigma'}^2- 2\re \ls \vx_0,\vz \rs}
= \frac{2(\beta_0 \|\vv\|_2^2-\beta_0\alpha \ls \vv,\vu \rs)}{ (\beta_0 -1+\alpha)^2+2(\beta_0 \|\vv\|_2^2-\beta_0\alpha \ls \vv,\vu \rs)}\\
&= \frac{2 \beta_0 (n-\alpha \text{tr}(\vu))}{ (\beta_0 -1+\alpha)^2 + 2 \beta_0 (n-\alpha \text{tr}(\vu))}
=  \frac{n-\alpha \text{tr}(\vu)}{ \tfrac{(\beta_0 -1+\alpha)^2}{2\beta_0} + (n-\alpha \text{tr}(\vu))}.
\end{split}
\]
Since $\beta_{0} \mapsto (\beta_0 -1+\alpha)^2/\beta_0 $ is minimized for $\beta_0= 1-\alpha$ and increasing for $\beta_0> 1-\alpha$, and since $1-\alpha u_{min} \geq 1-\alpha$, the above expression is minimized at $\beta_0 =1-\alpha u_{min}$ yielding 
\[
\delta_\Sigma(\vz) \geq \frac{n-\alpha \text{tr}(\vu)}{ \frac{(\alpha(1-u_{min}))^2}{2(1-\alpha u_{min})} + (n-\alpha \text{tr}(\vu))}= \frac{1}{\frac{1}{2} \cdot \frac{(\alpha(1-u_{min}))^2}{(1-\alpha u_{min})(n-\alpha \text{tr}(\vu))}  + 1} 
\]
We look for the supremum of $\frac{(\alpha(1-u_{min}))^2}{(1-\alpha u_{min})(n-\alpha \text{tr}(\vu))} $ over bi-stochastic $\vu$ and $0 \leq \alpha \leq 1$, which is obviously the supremum for $0<\alpha \leq 1$.  Considering $\gamma = 1/\alpha \in [1, +\infty)$ we observe that  
\[
\left(\frac{(\alpha(1-u_{min}))^2}{(1-\alpha u_{min})(n-\alpha \text{tr}(\vu))}\right)^{-1} = 
\frac{n(\gamma- u_{min})(\gamma- \text{tr}(\vu)/n)}{(1-u_{min})^2}
\]
is a second degree polynomial with roots $0 \leq u_{min}\leq  \text{tr}(\vu)/n \leq 1$. Thus this expression reaches its minimum at $\gamma =1$, where its value is 
\begin{equation}
\frac{n- \text{tr}(\vu)}{1-u_{min}} =  \frac{1 -u_{min} + n-1 - (\text{tr}(\vu)-u_{min})}{1-u_{min}} = 1 +\frac{ n-1 - (\text{tr}(\vu)-u_{min})}{1-u_{min}} \geq 1
\end{equation}
where we used the fact  that for all $i$, $u(i,i)\leq 1$ which implies $\text{tr}(\vu)-u_{min} \leq n-1$.  
Finally,
\begin{equation}
\delta_{\Sigma'}(\vz) \geq  \frac{1}{\frac{1}{2} \cdot \frac{1}{ \inf_{\gamma \geq 1}\frac{n(\gamma- u_{min})(\gamma- \text{tr}(\vu)/n)}{(1-u_{min})^2}}+1} \geq  \frac{1 }{ \tfrac{1}{2}+ 1} = \frac{2}{3}.
\end{equation}
Since this holds for any $\vz \in \sT_{\Sigma'}(\Sigma')$ we conclude that $\delta_{\Sigma'}(\|\cdot\|_{\Sigma'}) \geq 2/3$.
\end{proof}

\subsection{Limits of $\|\cdot\|_\Sigma$ : structured sparsity and low rank matrix recovery}\label{sec:sigma_norm_candidate}

In two examples --structured sparsity and low rank matrix recovery-- we now show that we cannot use $\|\cdot\|_\Sigma$ as a regularizer to perform uniform recovery. With the  set  $\sT_\Sigma(\Sigma)$, we can characterize when recovery with $f(\cdot) = \|\cdot\|_\Sigma$ in minimization~\eqref{eq:minimization} and at least one dimension-reducing linear operator $\mM$ is {\em impossible}: 
if $\sT_\Sigma(\Sigma) = \sE(\Sigma)$ then for any $\mM$ with a non-trivial null space we have $(\tker M \setminus {0}) \cap \sT_f(\Sigma) \neq \emptyset$ and recovery with $M$ is not possible. 

As we will see, this is the case for $k$-sparse vectors with $k \geq 2$, establishing that the norm $\|.\|_\Sigma$  considered in \cite{Argyriou_2012} (named $k$-support norm there)  does {\em not} permit uniform recovery of all such sparse vectors. We show the unit ball of this norm in Figure~\ref{fig:norm_sigma}. We observe that the span of the set of descent vectors at 1-sparse vectors are half-spaces, preventing any hope of recovery of these vectors. The norm $\|\cdot\|_\Sigma$ has also been considered in \cite{Richard_2014} for simultaneously sparse and low rank matrix recovery. The limit case of the low rank matrices shows that uniform recovery with $\|\cdot\|_\Sigma$ is not generally possible in this case either.
\begin{figure}[h]
\centering
\begin{tabular}{c}
\includegraphics[width=0.45\linewidth]{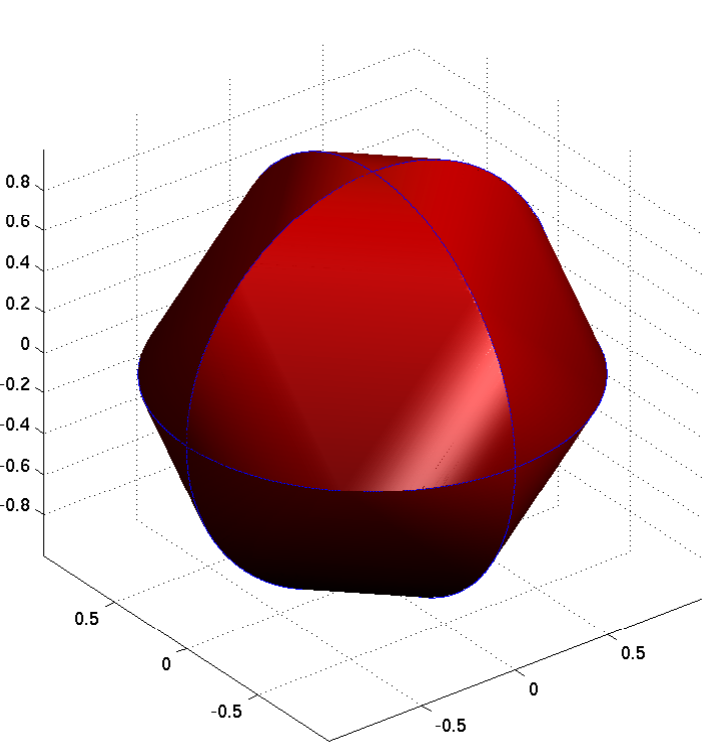}
\end{tabular}
\caption{The unit ball of the norm $\|\cdot\|_\Sigma$ for 2-sparse vectors in 3D (K-support norm).}\label{fig:norm_sigma}
\end{figure}

\begin{proposition}[Structured sparsity and the norm $\|\cdot\|_\Sigma$] \label{prop:sigm_norm_structured} 
 For $K\geq 2$,  uniform recovery of $K$-group sparse vectors with the norm $\|\cdot\|_\Sigma$ as a regularizer is impossible. 
\end{proposition}
\begin{proof}
If $K\geq 2$, we prove that $\sT_\Sigma(\Sigma) = \sE(\Sigma)$. By contradiction, suppose there is $\vz \in \sT_\Sigma(\Sigma)^c$. Consider a group $H = \{h\}$ and $\vx_h$ a normalized $1$-group sparse vector. We have $\|\vx_h\|_\Sigma = \|\vx_h\|_\sH=1$, and $\vx_h \in \Sigma_1 \subset \Sigma_K \subset \Sigma$. 
Since $\vz \in \sT_\Sigma(\Sigma)^c$ we have for all $\lambda \in \bR \setminus \{0\}$ :
\[
1 = \|\vx_h\|_\Sigma < \|\vx_h +\lambda \vz \|_\Sigma.
\]
We now upper bound  $\|\vx_h +\lambda \vz\|_\Sigma$ by exhibiting an atomic decomposition of $\vx_h + \lambda \vz$. Let $\vx_h +\lambda \vz= \vv_h +\sum_{g \in |H^c|} \vv_g$ be the canonical decomposition of $\vx_{h}+ \lambda \vz$ into $1$-group sparse vectors : $\vv_h,\vv_g \in \Sigma_1$.  Remark that $\vv_h = \vx_h +\lambda \vz_h$ and $\vv_g = \vz_g$. For each $g_i \in H^c$ (with $1\leq i\leq |H^c|$), let $\vu_i \in \Sigma_2 \subset \Sigma_K = \Sigma$ such that $\vu_i= \frac{1}{|H^c|} \vv_h + \vv_{g_i} $. Then $\vx_h +\lambda \vz = \sum_{i=1}^{|H^c|}\vu_i $  hence $ \|\vx_h +\lambda \vz \|_\Sigma \leq   \sum_{i=1}^{|H^c|} \|\vu_i\|_\sH$ and
\begin{equation}
\begin{split}
1 &<  \sum_{i=1}^{|H^c|} \|\vu_i\|_\sH 
=  \frac{1}{|H^c|} \sum_{i=1}^{|H^c|} \|\vx_h +\lambda \vz_h + \lambda |H^{c}| \vz_{g_i}) \|_\sH\\
&=   \frac{1}{|H^c|} \sum_{i=1}^{|H^c|} \sqrt{ \| \vx_h +\lambda \vz_h\|_\sH^2 + |\lambda|^2 |H^{c}|^{2} \|z_{g_i}\|_\sH^2}\\
&=   \frac{1}{|H^c|} \sum_{i=1}^{|H^c|} \sqrt{\|\vx_h\|_\sH^2 + 2\lambda \ls  \vx_h, \vz_h \rs +\lambda^2 (\|\vz_h\|_\sH^2 + |H^{c}|^{2} \|\vz_{g_i}\|_\sH^2)} \defin F(\lambda)
\end{split}
\end{equation}
We calculate the derivative $F'(0)$ : 
\begin{equation}
\begin{split}
F'(\lambda) &= -\frac{1}{|H^c|}\sum_{i=1}^{|H^c|} \frac{2 \ls \vx_h, \vz_h \rs +2 \lambda (\|\vz_h\|_\sH^2 + |H^{c}|^{2} \|\vz_{g_i}\|_\sH^2)}{2 \sqrt{ \| \vx_h\|_\sH^2 + 2\lambda \ls \vx_h, \vz_h \rs +\lambda^2 (\|\vz_h\|_\sH^2 + |H^{c}|^{2}\|\vz_{g_i}\|_\sH^2)}}\\ 
F'(0) &= - \frac{\ls \vx_h, \vz_h \rs }{ \|\vx_h\|_\sH}\\ 
\end{split}
\end{equation}
Besides, since $F(\lambda) > 1=F(0)$  for all $\lambda \neq 0$, we must have $F'(0)=0$, i.e. $\ls \vx_h, \vz_h \rs =0$. 

This being true for all normalized $1$-group sparse vectors with support $H = \{h\}$, $\vx_h$, we must have  $\vz = \vz_{H^c}$. This being true for all 1-group sparse supports $H$, implies $\vz = 0$. However $\vz=0$ is {\em not} in $\sT_\Sigma(\Sigma)^c$ therefore contradicting our assumption.  Consequently, $\sT_\Sigma(\Sigma)^c = \emptyset$ and $\sT_\Sigma(\Sigma) = \sE(\Sigma) $.
\end{proof}

The same fact holds with low rank matrices. 

\begin{proposition}[Low rank matrix recovery and the norm $\|\cdot\|_\Sigma$]
For $r \geq 2$, uniform recovery of $\Sigma_r$ (the set of matrices of rank at most $r$) with $\|\cdot\|_\Sigma$ is impossible. 
\end{proposition}
\begin{proof}
The proof is almost identical to that of Proposition~\ref{prop:sigm_norm_structured}, by using the singular value decomposition of $z$ and considering $x+z$ where $x$ is the rank-one matrix associated with the largest singular value of $z$.
\end{proof}

\section{Discussion and future work}\label{sec:discussion}

In this article, we gave a framework to assess the compatibility of a regularizer $f$ with a given generic model set $\Sigma$ (a cone). The number of examples that are treated show the flexibility and the effectiveness of this framework. Some additional questions emerge naturally from this work. 

\paragraph{Sharpness of $\delta_\Sigma(f)$} The sharpness of RIP constants can be considered in two ways : a weak sharpness or a strong sharpness. 

Given a model $\Sigma$ and a regularizer $f$, we say that the RIP constant $\delta_{\Sigma}^{strong}(f)$ has \emph{strong sharpness} if the RIP with $\delta <  \delta_{\Sigma}^{strong}(f)$ implies uniform recovery from $\Sigma$ with $f$, and there exists linear operators with RIP constant arbitrarily close to  $\delta_{\Sigma}^{strong}(f)$ for which uniform recovery from $\Sigma$ with $f$ is impossible. 

In the context of $K$-sparse recovery in dimension $N$, we consider a family of models $\Sigma_{K}^{N}$ and of regularizers $f_{N}(\cdot) = \|\cdot\|_{\ell^{N}_{1}}$. The constant  $\delta^{weak}$ has \emph{weak sharpness} over this family of models and regularizers if the RIP with $\delta <  \delta^{weak}$ implies uniform recovery from $\Sigma_{K}^{N}$ with $f_{N}$ for any $K$, $N$, and there exists dimensions $K,N$ and matrices with RIP constant arbitrarily close to  $\delta^{weak}$ for which uniform recovery from $\Sigma_{K}^{N}$ with $f_{N}$ is impossible. The notion of weak sharpness extends to many other families of models and regularizers, e.g. for $r$-rank recovery with the nuclear norm, or the recovery of permutation matrices with the gauge generated by the Birkhoff polytope of bi-stochastic matrices.  

For classical families of models and regularizers (sparse recovery with the $\ell^1$ norm and low-rank matrix recovery with the nuclear norm), as well as for structured sparsity and the associated mixed-norm, Theorem~\ref{th:sharp_RIP} and Theorem~\ref{th:robust_RIP_eucl} are weakly sharp. Indeed, we know that there exist RIP matrices with constant arbitrarily close to  $1/\sqrt{2}$ which do not permit uniform recovery \cite{Davies_2009}, and since $\delta_{\Sigma}(f) \geq 1/\sqrt{2}$ one must have $\delta^{weak} = 1/\sqrt{2}$. 
Thus for these families of models,  our results are weakly sharp in the sense that
\[\delta_\Sigma(f) \geq \delta^{weak} =  1/\sqrt{2}.\]


Considering block sparsity, further analysis would be needed to check weak shaprness $\delta_\Sigma(f)$. However one can already observe that it complies with the necessary dependency on the number of blocks $J$ \cite{Bastounis_2015}. 

The question of strong sharpness remains, i.e. do we have the pointwise equality (for {\em each} $f$ and $\Sigma$)
\[\delta_\Sigma(f)= \delta_{\Sigma}^{strong}(f)\ ?\]
A related open question for the family of sparse models is to determine whether the pointwise equality $\delta_{\Sigma_K^N}^{strong}(f)= \delta^{weak}$ holds. Answering this would close the discussion on RIP constants for such models.

\paragraph{Optimal decoders with respect to a model set $\Sigma$} Some new regularizer design appeared in Section \ref{sec:structured_sparsity}. We showed that adequately chosen weights for the block structured norm improve the RIP recovery guarantees for block structured sparse vectors. The admissible RIP constant $\delta_{\Sigma}(f)$ introduces a hierarchy in potential decoders that is also consistent with the quality of stability constants. From an algorithm design perspective, and given a constrained function class $f \in \mathcal{C}$, we can define :
\begin{equation}
\delta_\Sigma(\mathcal{C}) \defin \sup_{f \in \mathcal{C}} \delta_\Sigma(f)
\end{equation}
The design of regularizers under this perspective could be addressed by studying the following problems:
\begin{enumerate}
 \item Is $\delta_\Sigma(\mathcal{C}) > 0$?
 \item If yes, is the optimum reached at a function $f_0 \in \mathcal{C}$?
 \item Can we characterize $f_0$? Is it unique?
\end{enumerate}

\paragraph{Convex regularizer design}

It is particularly interesting to study the class $\mathcal{C}_{\text{convex}}$ of convex regularizers. We showed in Section~\ref{sec:definitions} that for (almost) any convex regularizer $f$ such that $\delta_\Sigma(f)\geq \delta$, there is an atomic norm $\|\cdot\|_\sA$ such that $\delta_\Sigma(\|\cdot\|_\sA) \geq \delta_\Sigma(f) \geq \delta$. Thus we can restrict the study to the class $\mathcal{C}_{\text{atomic}}$ of atomic norm regularizers $f$. 
Considering the function $\vz \mapsto \delta_\Sigma(\vz)$, we can define the set $T_\Sigma(\delta) \defin \{ \vz \in \sH : \delta_\Sigma(\vz) \geq  \delta\}$. Then $\delta_\Sigma(f) \geq \delta$ is equivalent to $\sT_{f} \subset T_\Sigma(\delta)$. This constraint on the descent set characterizes regularizers $f$ such that recovery with RIP $\delta$ is possible. In the case of 1-sparse vectors this condition can be characterized exactly and yields atomic norms that are close to the $\ell^1$ norm. When $\delta_\Sigma(f)\geq 1/\sqrt{2}$ is required, $f$ is necessarily a multiple of the $\ell^1$ norm. This route looks promising for the design of convex regularizers for model sets where such regularizers are unknown or sub-optimal.

\section*{Acknowledgments}
The authors would like to thank B. Adcock, M. Davies, N. Keriven, G. Kutyniok , G. Obozinski and the anonymous reviewers for their insightful comments which helped us improve this paper at various stages of its preparation. 

\section*{References}
\bibliographystyle{abbrv}
 \bibliography{yanntraonmilin}

\begin{thebibliography}{10}

\bibitem{Adcock_2013b}
B.~Adcock, A.~Hansen, B.~Roman, and G.~Teschke.
\newblock Generalized sampling: stable reconstructions, inverse problems and
  compressed sensing over the continuum.
\newblock {\em arXiv preprint arXiv:1310.1141}, 2013.

\bibitem{Adcock_2012}
B.~Adcock and A.~C. Hansen.
\newblock {A generalized sampling theorem for stable reconstructions in
  arbitrary bases}.
\newblock {\em Journal of Fourier Analysis and Applications}, 18(4):685--716,
  2012.

\bibitem{Adcock_2015}
B.~Adcock and A.~C. Hansen.
\newblock {Generalized sampling and infinite dimensional compressed sensing}.
\newblock {\em Foundations of Computational Mathematics (to appear)}, 2015.

\bibitem{Adcock_2013}
B.~Adcock, A.~C. Hansen, and C.~Poon.
\newblock {Beyond Consistent Reconstructions: Optimality and Sharp Bounds for
  Generalized Sampling, and Application to the Uniform Resampling Problem}.
\newblock {\em SIAM Journal on Mathematical Analysis}, 45(5):3132--3167, 2013.

\bibitem{Adcock_2013c}
B.~Adcock, A.~C. Hansen, C.~Poon, and B.~Roman.
\newblock Breaking the coherence barrier: A new theory for compressed sensing.
\newblock {\em arXiv preprint arXiv:1302.0561}, 2013.

\bibitem{Charalambos_2006}
C.~D. Aliprantis and K.~C. Border.
\newblock {\em Infinite dimensional analysis, A Hitchhiker’s Guide}.
\newblock Springer, 2006.

\bibitem{Ameluxen_2014}
D.~Amelunxen, M.~Lotz, M.~B. McCoy, and J.~A. Tropp.
\newblock Living on the edge: Phase transitions in convex programs with random
  data.
\newblock {\em Information and Inference}, page iau005, 2014.

\bibitem{Argyriou_2012}
A.~Argyriou, R.~Foygel, and N.~Srebro.
\newblock {Sparse Prediction with the k-Support Norm}.
\newblock In F.~Pereira, C.~J.~C. Burges, L.~Bottou, and K.~Q. Weinberger,
  editors, {\em Advances in Neural Information Processing Systems 25}, pages
  1457--1465. Curran Associates, Inc., 2012.

\bibitem{Ayaz_2014}
U.~Ayaz, S.~Dirksen, and H.~Rauhut.
\newblock Uniform recovery of fusion frame structured sparse signals.
\newblock {\em Applied and Computational Harmonic Analysis}, 2016.

\bibitem{Baraniuk_2010}
R.~G. Baraniuk, V.~Cevher, M.~F. Duarte, and C.~Hegde.
\newblock Model-based compressive sensing.
\newblock {\em Information Theory, IEEE Transactions on}, 56(4):1982--2001,
  2010.

\bibitem{Bastounis_2015}
A.~Bastounis and A.~C. Hansen.
\newblock On random and deterministic compressed sensing and the restricted
  isometry property in levels.
\newblock In {\em Sampling Theory and Applications (SampTA), 2015 International
  Conference on}, pages 297--301. IEEE, 2015.

\bibitem{Blumensath_2011}
T.~Blumensath.
\newblock Sampling and reconstructing signals from a union of linear subspaces.
\newblock {\em Information Theory, IEEE Transactions on}, 57(7):4660--4671,
  2011.

\bibitem{Blumensath_2009}
T.~Blumensath and M.~E. Davies.
\newblock {Sampling Theorems for Signals From the Union of Finite-Dimensional
  Linear Subspaces}.
\newblock {\em Information Theory, IEEE Transactions on}, 55(4):1872--1882,
  Apr. 2009.

\bibitem{Bonsall_1991}
F.~F. Bonsall.
\newblock A general atomic decomposition theorem and banach's closed range
  theorem.
\newblock {\em The Quarterly Journal of Mathematics}, 42(1):9--14, 1991.

\bibitem{Bourrier_2014}
A.~Bourrier, M.~Davies, T.~Peleg, P.~Perez, and R.~Gribonval.
\newblock Fundamental performance limits for ideal decoders in high-dimensional
  linear inverse problems.
\newblock {\em Information Theory, IEEE Transactions on}, 60(12):7928--7946,
  Dec 2014.

\bibitem{Cai_2014}
T.~Cai and A.~Zhang.
\newblock Sparse representation of a polytope and recovery of sparse signals
  and low-rank matrices.
\newblock {\em Information Theory, IEEE Transactions on}, 60(1):122--132, Jan
  2014.

\bibitem{Candes_2008}
E.~J. Cand\`{e}s.
\newblock {The restricted isometry property and its implications for compressed
  sensing}.
\newblock {\em Comptes Rendus Mathematique}, 346(9-10):589--592, May 2008.

\bibitem{Candes_2011b}
E.~J. Candes, Y.~C. Eldar, D.~Needell, and P.~Randall.
\newblock Compressed sensing with coherent and redundant dictionaries.
\newblock {\em Applied and Computational Harmonic Analysis}, 31(1):59--73,
  2011.

\bibitem{Candes_2010}
E.~J. Candes and Y.~Plan.
\newblock {Matrix Completion With Noise}.
\newblock {\em Proceedings of the IEEE}, 98(6):925--936, June 2010.

\bibitem{Candes_2006b}
E.~J. Cand\`{e}s, J.~Romberg, and T.~Tao.
\newblock {Robust uncertainty principles: exact signal reconstruction from
  highly incomplete frequency information}.
\newblock {\em Information Theory, IEEE Transactions on}, 52(2):489--509, Feb.
  2006.

\bibitem{Chandrasekaran_2012}
V.~Chandrasekaran, B.~Recht, P.~Parrilo, and A.~Willsky.
\newblock The convex geometry of linear inverse problems.
\newblock {\em Foundations of Computational Mathematics}, 12(6):805--849, 2012.

\bibitem{Davies_2009}
M.~E. Davies and R.~Gribonval.
\newblock Restricted isometry constants where $\ell^p$ sparse recovery can fail
  for $0 < p \leq 1$.
\newblock {\em Information Theory, IEEE Transactions on}, 55(5):2203--2214,
  2009.

\bibitem{Dirksen_2014}
S.~Dirksen.
\newblock Dimensionality reduction with subgaussian matrices: a unified theory.
\newblock {\em Foundations of Computational Mathematics}, pages 1--30, 2015.

\bibitem{Donoho_2006}
D.~L. Donoho.
\newblock {For most large underdetermined systems of linear equations the
  minimal $\ell_1$-norm solution is also the sparsest solution}.
\newblock {\em Comm. Pure Appl. Math.}, 59(6):797--829, June 2006.

\bibitem{Eftekhari_2015}
A.~Eftekhari and M.~B. Wakin.
\newblock New analysis of manifold embeddings and signal recovery from
  compressive measurements.
\newblock {\em Applied and Computational Harmonic Analysis}, 39(1):67--109,
  2015.

\bibitem{Eldar_2010}
Y.~C. Eldar, P.~Kuppinger, and H.~Bolcskei.
\newblock Block-sparse signals: Uncertainty relations and efficient recovery.
\newblock {\em Signal Processing, IEEE Transactions on}, 58(6):3042--3054,
  2010.

\bibitem{Eldar_2009}
Y.~C. Eldar and M.~Mishali.
\newblock {Robust Recovery of Signals From a Structured Union of Subspaces}.
\newblock {\em Information Theory, IEEE Transactions on}, 55(11):5302--5316,
  Nov. 2009.

\bibitem{Feichtinger_2002}
H.~G. Feichtinger and G.~Zimmermann.
\newblock An exotic minimal banach space of functions.
\newblock {\em Mathematische Nachrichten}, 239-240(1):42--61, 2002.

\bibitem{Foucart_2013}
S.~Foucart and H.~Rauhut.
\newblock {\em A mathematical introduction to compressive sensing}.
\newblock Springer, 2013.

\bibitem{Gribonval_2008}
R.~Gribonval and M.~Nielsen.
\newblock Beyond sparsity: Recovering structured representations by ${\ell}^1$
  minimization and greedy algorithms.
\newblock {\em Advances in Computational Mathematics}, 28(1):23--41, 2008.

\bibitem{Horn_1990}
R.~A. Horn and C.~R. Johnson.
\newblock {\em Matrix analysis}.
\newblock Cambridge university press, 1990.

\bibitem{Johnson_1984}
W.~B. Johnson and J.~Lindenstrauss.
\newblock Extensions of lipschitz mappings into a hilbert space.
\newblock {\em Contemporary mathematics}, 26(189-206):1, 1984.

\bibitem{Puy_2015}
G.~Puy, M.~E. Davies, and R.~Gribonval.
\newblock Recipes for stable linear embeddings from hilbert spaces to
  $\mathbb{R}^m$.
\newblock {\em arXiv preprint arXiv:1509.06947}, 2015.

\bibitem{Recht_2010}
B.~Recht, M.~Fazel, and P.~Parrilo.
\newblock {Guaranteed Minimum-Rank Solutions of Linear Matrix Equations via
  Nuclear Norm Minimization}.
\newblock {\em SIAM Review}, 52(3):471--501, 2010.

\bibitem{Richard_2014}
E.~Richard, G.~R. Obozinski, and J.-P. Vert.
\newblock Tight convex relaxations for sparse matrix factorization.
\newblock In Z.~Ghahramani, M.~Welling, C.~Cortes, N.~Lawrence, and
  K.~Weinberger, editors, {\em Advances in Neural Information Processing
  Systems 27}, pages 3284--3292. Curran Associates, Inc., 2014.

\bibitem{Robinson_2010}
J.~C. Robinson.
\newblock {\em Dimensions, embeddings, and attractors}, volume 186.
\newblock Cambridge University Press, 2010.

\bibitem{Rockafellar_1970}
R.~T. Rockafellar.
\newblock {\em Convex analysis}.
\newblock Number~28. Princeton university press, 1970.

\bibitem{Rockafellar_2009}
R.~T. Rockafellar and R.~J.-B. Wets.
\newblock {\em Variational analysis}, volume 317.
\newblock Springer Science \& Business Media, 2009.

\bibitem{Studer_2013}
C.~Studer and R.~G. Baraniuk.
\newblock {Stable restoration and separation of approximately sparse signals}.
\newblock {\em Applied and Computational Harmonic Analysis}, Sept. 2013.

\bibitem{Traonmilin_2015}
Y.~Traonmilin, S.~Ladjal, and A.~Almansa.
\newblock Robust multi-image processing with optimal sparse regularization.
\newblock {\em Journal of Mathematical Imaging and Vision}, 51(3):413--429,
  2015.

\bibitem{Yuan_2006}
M.~Yuan and Y.~Lin.
\newblock Model selection and estimation in regression with grouped variables.
\newblock {\em Journal of the Royal Statistical Society: Series B (Statistical
  Methodology)}, 68(1):49--67, 2006.

\end{thebibliography}
\end{document}